\documentclass[a4paper,11pt]{article}
 \usepackage[]{geometry}
 \usepackage[utf8]{inputenc}
 \usepackage[parfill]{parskip}
 \usepackage{graphicx}	
 \usepackage{caption}
 \usepackage{subcaption}		
 \usepackage{amssymb}
 \usepackage{amsmath}
 \usepackage{amsthm}
 \usepackage{natbib}
 \usepackage{lscape}
 \usepackage{longtable}
 \usepackage[flushleft]{threeparttable}
 \usepackage{tabularx}
 \usepackage{booktabs}
 \usepackage{multirow}
 \usepackage{bbm}
 \usepackage{xcolor}
 \usepackage[title]{appendix}
 \usepackage[colorlinks,citecolor=blue,urlcolor=blue,bookmarks=false,hypertexnames=true]{hyperref}

 \newcommand\numberthis{\addtocounter{equation}{1}\tag{\theequation}}%add number to equation in align*
 \usepackage{tikz}
 \usetikzlibrary{arrows}
 \usepackage{enumitem}

 \newtheorem{assume}{Assumption}
 
 \newtheorem{prop}{Proposition} % small theorem
 
 \newtheorem{theorem}{Theorem} % main result
 
 \newtheorem{remark}{Remark}

 \title{Individual Causal Inference\\Using Panel Data With Multiple Outcomes}
 \author{Wei Tian\\UNSW
 }
\date{August 24, 2021}

\begin{document}
 \maketitle

 \begin{abstract}
   Policy evaluation in empirical microeconomics has been focusing on estimating the average treatment effect and more recently the heterogeneous treatment effects, often relying on the unconfoundedness assumption. We propose a method based on the interactive fixed effects model to estimate treatment effects at the individual level, which allows both the treatment assignment and the potential outcomes to be correlated with the unobserved individual characteristics. This method is suitable for panel datasets where multiple related outcomes are observed for a large number of individuals over a small number of time periods. Monte Carlo simulations show that our method outperforms related methods. To illustrate our method, we provide an example of estimating the effect of health insurance coverage on individual usage of hospital emergency departments using the Oregon Health Insurance Experiment data.
   We find heterogeneous treatment effects in the sample. Comparisons between different groups show that the individuals who would have fewer emergency-department visits if covered by health insurance were younger and not in very bad physical conditions. However, their access to primary care were limited due to being in much more disadvantaged positions financially, which made them resort to using the emergency department as the usual place for medical care.
   Health insurance coverage might have decreased emergency-department use among this group by increasing access to primary care and possibly leading to improved health.
   In contrast, the individuals who would have more emergency-department visits if covered by health insurance were more likely to be older and in poor health. So even with access to primary care, they still used emergency departments more often for severe conditions, although sometimes for primary care treatable and non-emergent conditions as well.
   Health insurance coverage might have increased their emergency-department use by reducing the out-of-pocket cost of the visits.
 \end{abstract}

\section{Introduction}

The main focus of the policy evaluation literature has been the average treatment effect and more recently the heterogeneous treatment effects or conditional average treatment effects, which are the average treatment effects for heterogeneous subgroups defined by the observed covariates (for reviews of these methods, see \citealp*{athey2017state,abadie2018econometric}).
Ubiquitous in these studies is the unconfoundedness assumption, or the strong ignorability assumption, which requires all the covariates correlated with both the potential outcomes and the treatment assignment to be observed \citep{rosenbaum1983central}.\footnote{This is also known as selection on observables or the conditional independence assumption.}
Under this assumption, the potential outcomes and the treatment status are independent conditional on the observed covariates, and the difference between the mean outcomes of the treated and the untreated groups with the same values of the observed covariates is an unbiased estimator of the average treatment effect for the units in the groups.
The unconfoundedness assumption is satisfied in randomised controlled experiments, but may not be plausible otherwise even with a rich set of covariates, since the access to certain essential individual characteristics remains limited for the researchers due to privacy or ethical concerns, despite the explosive growth of data availability in the big data era.

One popular method to circumvent the unconfoundedness assumption is difference-in-differences (DID), which assumes that the effect of the unobserved confounder on the untreated potential outcome is constant over time, so that the average outcomes of the treated and untreated units would follow parallel trends in the absence of the treatment.\footnote{Alternative methods that do not rely on the unconfoundedness assumption include the instrumental variables method and the regression discontinuity design, which estimate the average treatment effect for specific subpopulations (the compliers or those with values of the running variable near the cutoff).} This is also a strong assumption, and in many cases is not supported by data.
The interactive fixed effects model relaxes the ``parallel trends'' assumption and allows the unobserved confounders to have time-varying effects on the outcomes, by modeling them using an interactive fixed effects term, which incorporates the additive unit and time fixed effects model or difference-in-differences as a special case \citep{bai2009panel}.

Several methods have been developed based on the interactive fixed effects model to estimate the treatment effect on a single or several treated units, where the units are observed over an extended period of time before the treatment \citep{abadie2010synthetic,hsiao2012panel,xu2017generalized}. These methods exploit the cross-sectional correlations attributed to the unobserved common factors to predict the counterfactual outcomes for the treated units, and are mainly used in macroeconomic settings with a large number of pretreatment periods, which is crucial for the results to be credible. For example, \cite{abadie2015comparative} point out that ``the applicability of the method requires a sizable number of preintervention periods'' and that ``we do not recommend using this method when the pretreatment fit is poor or the number of pretreatment periods is small'', while \cite{xu2017generalized} states that users should be cautious when there are fewer than 10 pretreatment periods.
As a consequence, despite the potential to estimate individual treatment effects without imposing the unconfoundedness assumption, these methods have not seen much use in empirical microeconomics, since the individuals are rarely tracked for more than a few periods that justify the use of these methods.

The main contribution of this paper is that we propose a method for estimating the individual treatment effects in applied microeconomic settings, characterised by multiple related outcomes being observed for a large number of individuals over a small number of time periods.
The method is based on the interactive fixed effects model, which assumes that an outcome of interest can be well approximated by a linear combination of a small number of observed and unobserved individual characteristics.
Analogous to \cite{hsiao2012panel} who predict the posttreatment outcomes using pretreatment outcomes in lieu of the unobserved time factors, we use a subsample of the pretreatment outcomes to replace the unobserved individual characteristics in the models, and use the remaining pretreatment outcomes as instrumental variables.
Although our method does not require a large number of pretreatment periods, the number of pretreatment outcomes needs to be at least as large as the number of unobserved individual characteristics, which may still be difficult to satisfy in microeconomic datasets if we use only a single outcome, especially if the treatment assignment took place in the early stages of the survey or if the study subjects are children or youths.
Utilising multiple related outcomes allows our method to be applicable in cases where there is only a single period before the treatment.
Under the assumption that these outcomes depend on roughly the same set of observed covariates and unobserved individual characteristics with time-varying and outcome-specific coefficients shared by all individuals, our method exploits the correlations across related outcomes and over time, which are induced by the unobserved individual characteristics, to predict the counterfactual outcomes and estimate the treatment effects for each individual in the posttreatment periods.

Our method has several advantages.
First, with the assumption on the model specification, it relaxes the arguably much stronger unconfoundedness assumption, and allows the treatment assignment to be correlated with the unobserved individual characteristics.
Second, it enables the estimation of treatment effects on the individual level, which may be helpful for designing more individualised policies to maximize social welfare, as well as for other fields such as precision medicine and individualised marketing. It also has the potential to be combined with more flexible machine learning methods to work with big datasets and more general nonlinear function forms.
Third, it is intuitive. In real life, we may never know a person through and through, and a viable approach to predicting the outcome of a person is using his or her related outcomes in the past, assuming that the outcomes are affected by the underlying individual characteristics and that these characteristics are stable over time, at least within the study period. For example, past academic performance is an important consideration when recruiting a student into college, as it is believed that a student that excelled in the past is likely to continue to have outstanding performance. To the extent that we may never observe all the confounders, this is perhaps the only way to predict potential outcomes and estimate treatment effects on the individual level in social sciences without going deeper to the levels of neuroscience or biology.
Fourth, our method has wide applicability, as it is common to have multiple related outcomes collected in microeconomics data. For example, we may observe several health related outcomes such as health facility usage, health related cost, general health, etc.

The rest of the study is organised as follows.
Section \ref{Ch3_sec_theory} presents the theoretical framework.
Section \ref{Ch3_sec_sim} examines the small sample performance of our method using Monte Carlo simulation, and compares it with related methods.
Section \ref{Ch3_sec_app} provides an empirical example of estimating the effect of health insurance coverage on individual usage of hospital emergency departments using the Oregon Health Insurance Experiment data.
Section \ref{Ch3_sec_con} concludes and discusses potential directions for future research. The proofs are collected in the appendix.

\section{Theory}\label{Ch3_sec_theory}

\subsection{Set Up}

Suppose that we observe $K$ outcomes in domain $\mathcal{K}=\{1,2,\dots,K\}$ for $N$ individuals or units over $T\ge 2$ time periods, where a domain refers to a collection of related outcomes that depend on the same set of observed covariates and unobserved characteristics. For example, health-related outcomes may be affected by observed covariates such as age, education, occupation and income, as well as unobserved individual characteristics such as genetic inheritance, health habits and risk preferences.
Assume that the $N_1$ individuals in the treated group $\mathcal{T}$ receive the treatment at period $T_0+1\le T$ and remain treated afterwards, while the $N_0=N-N_1$ individuals in the control group $\mathcal{C}$ remain untreated throughout the $T$ periods.
Denoting the binary treatment status for individual $i$ at time $t$ as $D_{it}$, we have $D_{it}=1$ for $i\in\mathcal{T}$ and $t>T_0$, and $D_{it}=0$ otherwise.

Following the ``Rubin Causal Model'' \citep{rubin1974estimating}, the treatment effect on outcome $k\in\mathcal{K}$ for individual $i$ at time $t$ is given by the difference between the treated and untreated potential outcomes
\begin{align}\label{Ch3_eq_ITE}
  \tau_{it,k}=Y_{it,k}^1-Y_{it,k}^0,
\end{align}
where $Y_{it,k}^1$ is the treated potential outcome, the outcome that we would observe for individual $i$ at time $t$ if $D_{it}=1$, and $Y_{it,k}^0$ is the untreated potential outcome, the outcome that we would observe if $D_{it}=0$.
% Since we only observe one of the potential outcomes for individual $i$ at any time $t$, depending on the treatment status $D_{it}$, i.e., the observed outcome is given by
% \begin{align}
%   Y_{it,k}=D_{it}Y_{it,k}^1+\left(1-D_{it}\right)Y_{it,k}^0,
% \end{align}
% we need to predict the counterfactual outcome.
% Therefore, the problem of estimating the individual treatment effects becomes the problem of predicting the potential outcomes for the individuals.
Instead of assuming the unconfoundedness condition, we characterise the two potential outcomes for individual $i$ at time $t$ and $k\in\mathcal{K}$ using the interactive fixed effects models:
\begin{align}
  Y_{it,k}^1= & \boldsymbol{X}_{it}'\boldsymbol{\beta}_{t,k}^1+\boldsymbol{\mu}_i'\boldsymbol{\lambda}_{t,k}^1+\varepsilon_{it,k}^1,\label{Ch3_eq_IFE1} \\
  Y_{it,k}^0= & \boldsymbol{X}_{it}'\boldsymbol{\beta}_{t,k}^0+\boldsymbol{\mu}_i'\boldsymbol{\lambda}_{t,k}^0+\varepsilon_{it,k}^0,\label{Ch3_eq_IFE0}
\end{align}
where $\boldsymbol{X}_{it}$ is the $r\times 1$ vector of observed covariates unaffected by the treatment, $\boldsymbol{\mu}_i$ is the $f\times 1$ vector of unobserved individual characteristics, $\boldsymbol{\beta}_{t,k}^1$ and $\boldsymbol{\lambda}_{t,k}^1$ are the $r\times 1$ and $f\times 1$ vectors of coefficients of $\boldsymbol{X}_{it}$ and $\boldsymbol{\mu}_i$ respectively for the treated potential outcome, $\boldsymbol{\beta}_{t,k}^0$ and $\boldsymbol{\lambda}_{t,k}^0$ are the coefficients for the untreated potential outcome, and $\varepsilon_{it,k}^1$ and $\varepsilon_{it,k}^0$ are the idiosyncratic shocks.

% \begin{remark}
%   \normalfont
%   The models in \eqref{Ch3_eq_IFE1} and \eqref{Ch3_eq_IFE0} together indicate that the treatment effects are deterministic functions of the underlying predictors. This is more general than assuming that the treatment effects are constant. A similar assumption is discussed in \cite{athey2021matrix}, where the treatment effect is assumed to have a low-rank pattern.
%   To estimate the treatment effects on the treated, we only need to assume the functional form in \eqref{Ch3_eq_IFE0} for the untreated potential outcomes and that $Y_{it,k}^1=Y_{it,k}^0+\tau_{it,k}$, where $\tau_{it,k}$ is treated as fixed given the sample, as in \cite{abadie2010synthetic}, \cite{hsiao2012panel} and \cite{xu2017generalized}.
%   The model specification for the treated potential outcome in \eqref{Ch3_eq_IFE1} is needed to estimate the treatment effects for the untreated individuals.
% \end{remark}

\begin{remark}
  \normalfont
  Our models for the potential outcomes are quite general, and incorporate the models in \cite{abadie2010synthetic}, \cite{hsiao2012panel} and \cite{xu2017generalized}, as well as the the additive fixed effects model for difference-in-differences as special cases.\footnote{Specifically, if we assume $\boldsymbol{\beta}_{t,k}^0=\boldsymbol{\beta}_{k}^0$, model \eqref{Ch3_eq_IFE0} reduces to the model in \cite{bai2009panel} and \cite{xu2017generalized}; if we assume $\boldsymbol{X}_{it}=\boldsymbol{X}_{i}$ and the first element of $\boldsymbol{\mu}_i$ is 1, model \eqref{Ch3_eq_IFE0} reduces to the model in \cite{abadie2010synthetic}; if we assume $\boldsymbol{X}_{it}=\boldsymbol{X}_{i}$ are unobserved and the first element of $\boldsymbol{\lambda}_{t,k}^0$ is 1, then model \eqref{Ch3_eq_IFE0} reduces to the model in \cite{hsiao2012panel}; if we assume $\boldsymbol{\mu}_i=[1\enspace a_i]'$ and $\boldsymbol{\lambda}_{t,k}^0=[b_t\enspace 1]'$, then model \eqref{Ch3_eq_IFE0} reduces to the additive fixed effects model for difference-in-differences.}
  Note that the related outcomes need not depend on exactly the same set of observed covariates and unobserved individual characteristics. The vectors of outcome-specific and time-varying coefficients may contain 0, so that outcome $k$ may be affected by some of the observed covariates and unobserved individual characteristics in some periods, but not necessarily by all of them in all periods, as long as there is enough variation in the coefficients over time or across the outcomes. The potential outcomes are also allowed to depend on predictors not included in $\boldsymbol{X}_{it}$ or $\boldsymbol{\mu}_i$, as long as they are not correlated with the included predictors and the treatment status so that they can be treated as part of the idiosyncratic shock.
\end{remark}

\medskip

The regularity conditions on the observed covariates and the unobserved individual characteristics are stated in Assumption \ref{Ch3_assume_predictor}, and the assumptions on the idiosyncratic shocks are given in Assumption \ref{Ch3_assume_error}.

\begin{assume}[]\label{Ch3_assume_predictor}
  \leavevmode
  \begin{enumerate}[label=\arabic*)]
    \item $\boldsymbol{X}_{it}$, $\boldsymbol{\mu}_i$ are independent for all $i$, and are identically distributed for all $i\in\mathcal{T}$ and all $i\in\mathcal{C}$ respectively;
    \item There exists $M\in[0,\infty)$ such that $\mathbb{E}\Vert\boldsymbol{X}_{it}\Vert^4<M$ and $\mathbb{E}\Vert\boldsymbol{\mu}_{i}\Vert^4<M$.
  \end{enumerate}
\end{assume}

\begin{assume}[]\label{Ch3_assume_error}
  For $d\in \{0,1\}$, we have
  \leavevmode
  \begin{enumerate}[label=\arabic*)]
    \item $\mathbb{E}\left(\varepsilon_{it,k}^d\mid \boldsymbol{X}_{js},\boldsymbol{\mu}_j,D_{js}\right)=0$ for all $i$, $j$, $t$, $s$ and $k$;
    \item $\varepsilon_{it,k}^d$ are independent across $i$ and $t$;
    \item $\mathbb{E}\left(\varepsilon_{it,k}^d,\varepsilon_{it,l}^d\right)=\sigma_{t,kl}^d$ for all $i$, $t$, $k$, $l$;
    \item There exists $M\in[0,\infty)$ such that $\mathbb{E}|\varepsilon_{it,k}^d|^4<M$ for all $i$, $t$, $k$.
  \end{enumerate}
\end{assume}

\begin{remark}
  \normalfont
  The distributions of the observed covariates and the unobserved individual characteristics are allowed to differ for the treated and untreated individuals, i.e., selection on unobservables is allowed, which is a great advantage over the policy evaluation methods that rely on the unconfoundedness condition.
  The idiosyncratic shocks are assumed to have zero mean conditional on the observed covariates, unobserved individual characteristics and the treatment status. They are also assumed to be independent across individuals and time periods, as the unobserved interactive fixed effects that account for the cross-sectional and time-serial correlations have been separated out.\footnote{The idiosyncratic shocks may be allowed to be correlated both over time and across outcomes, as long as they can be modelled parametrically and removed using a quasi-differencing approach. This is left for future research.} Furthermore, they are assumed to be homoskedastic across individuals for our inference method to be valid. The last part in Assumption \ref{Ch3_assume_error} is a regularity condition which, together with the conditions in Assumption \ref{Ch3_assume_predictor}, ensures the weak law of large numbers and the central limit theorem hold.
\end{remark}

\medskip

Given the models for $Y_{it,k}^1$ and $Y_{it,k}^0$ in \eqref{Ch3_eq_IFE1} and \eqref{Ch3_eq_IFE0}, the individual treatment effect is identified by the observed covariates and the unobserved individual characteristics, i.e., two persons with the same values for these underlying predictors have identical individual treatment effect. Denote the set of observed covariates and unobserved individual characteristics as $\boldsymbol{H}_{it}=\left[\boldsymbol{X}_{it}'\enspace\boldsymbol{\mu}_i'\right]'$,
% $\boldsymbol{H}_i=\left[\boldsymbol{X}_{i1}'\enspace\cdots\enspace\boldsymbol{X}_{iT}'\enspace\boldsymbol{\mu}_i'\right]'$ 
% $m_{t,k}^1\left(\boldsymbol{h}\right)=\mathbb{E}\left(Y_{it,k}^1\mid \boldsymbol{H}_{it}=\boldsymbol{h}\right)$ and $m_{t,k}^0\left(\boldsymbol{h}\right)=\mathbb{E}\left(Y_{it,k}^0\mid \boldsymbol{H}_{it}=\boldsymbol{h}\right)$, 
then the individual treatment effect for individual $i$ with $\boldsymbol{H}_{it}=\boldsymbol{h}_{it}$ is given by
\begin{align}\label{Ch3_eq_ITE1}
  \bar{\tau}_{it,k}\equiv\mathbb{E}\left(Y_{it,k}^1-Y_{it,k}^0\mid \boldsymbol{H}_{it}=\boldsymbol{h}_{it}\right),
\end{align}
which may appear similar to the conditional average treatment effect, but is different by conditioning not only on the observed covariates, but also on the unobserved individual characteristics.\footnote{As we assume the parametric models for the potential outcomes in \eqref{Ch3_eq_IFE1} and \eqref{Ch3_eq_IFE0} for all individuals, there is no need to impose additional assumptions on the propensity distribution for the individual treatment effect to be identified on its full support.}
% as in studies that estimate conditional average treatment effects nonparametrically

Our goal is to estimate the individual treatment effects $\bar{\tau}_{it,k}$, $i=1,\dots,N$. Once we have estimated the individual treatment effects, the estimates of the average treatment effects for heterogeneous subgroups defined by some observed covariates, also known as the conditional average treatment effects, and the estimate for the average treatment effect for the sample or the population are also readily available using the average of the estimated individual treatment effects in the corresponding groups.

As $\boldsymbol{\mu}_i$ is not observed, a direct application of least squares estimation to estimate the models in \eqref{Ch3_eq_IFE1} and \eqref{Ch3_eq_IFE0} would suffer from omitted variables bias. Since we have multiple outcomes that depend on the same set of underlying predictors, and we observe the untreated potential outcomes for all individuals prior to the treatment, we can use these pretreatment outcomes to replace $\boldsymbol{\mu}_i$ in the models.\footnote{This is analogous to the first step of the approach in \cite{hsiao2012panel}, who predict the posttreatment outcomes using pretreatment outcomes in lieu of the unobserved time factors in a small $N$, big $T$ environment.}
Stacking the $K$ outcomes observed in $t\le T_0$, we have
\begin{align}
  \boldsymbol{Y}_{it}^0=\boldsymbol{\beta}_{t}^0\boldsymbol{X}_{it}+\boldsymbol{\lambda}_{t}^0\boldsymbol{\mu}_i+\boldsymbol{\varepsilon}_{it}^0,
\end{align}
where $\boldsymbol{Y}_{it}^0$ and $\boldsymbol{\varepsilon}_{it}^0$ are $K\times 1$, $\boldsymbol{\beta}_{t}^0$ is $K\times r$, and $\boldsymbol{\lambda}_{t}^0$ is $K\times f$.
Let $\mathcal{P}\subseteq\left\{1,\cdots,T_0\right\}$ be a set of $P$ pretreatment periods. We can further stack the outcomes over these periods to get
\begin{align}\label{Ch3_eq_IFE_stack}
  \boldsymbol{Y}_i^\mathcal{P}=\boldsymbol{\delta}_i^\mathcal{P}+\boldsymbol{\lambda}^\mathcal{P}\boldsymbol{\mu}_i+\boldsymbol{\varepsilon}_i^\mathcal{P},
\end{align}
where $\boldsymbol{\delta}_i^\mathcal{P}=\left[\cdots \enspace \left(\boldsymbol{\beta}_{s}^0\boldsymbol{X}_{is}\right)' \enspace \cdots\right]'$ with $s\in\mathcal{P}$ is $KP\times 1$, $\boldsymbol{\lambda}^\mathcal{P}$ is $KP\times f$, and $\boldsymbol{\varepsilon}_i^\mathcal{P}$ is $KP\times 1$.

\medskip

To be able to recover $\boldsymbol{\mu}_i$ from the covariates and outcomes observed in $\mathcal{P}$, we need the following full rank condition, which ensures that there is enough variation in the effects of the unobserved individual characteristics over time or across different outcomes.

\begin{assume}[]\label{Ch3_assume_rank}
  ${\boldsymbol{\lambda}^\mathcal{P}}'\boldsymbol{\lambda}^\mathcal{P}$ has rank $f$.
\end{assume}

\begin{remark}
  \normalfont
  Although we do not require the number of pretreatment outcomes to be large, Assumption \ref{Ch3_assume_rank} implies that $KP$ needs to be at least as large as $f$.
  As $T_0$ (and thus $P$) is usually small in empirical microeconomics, this assumption is made more plausible by having $K>1$, i.e., using multiple related outcomes.
\end{remark}

\begin{remark}
  \normalfont
  The number of factors $f$ is generally not observed. To determine $f$, one may use the method in \cite{bai2002determining} when both $N$ and $T$ are large. One may also adopt a cross-validation procedure to choose $f$ that minimises the out-of-sample mean squared prediction error, as in \cite{xu2017generalized}. Although we do not estimate the interactive fixed effects term directly, we may choose the number of pretreatment outcomes that best accommodates $f$ using cross-validation as well, which will be discussed in more details later.
\end{remark}

\medskip

Under Assumption \ref{Ch3_assume_rank}, we can pre-multiply both sides of equation \eqref{Ch3_eq_IFE_stack} by $({\boldsymbol{\lambda}^\mathcal{P}}'\boldsymbol{\lambda}^\mathcal{P})^{-1}{\boldsymbol{\lambda}^\mathcal{P}}'$ to obtain
\begin{align}\label{Ch3_eq_mu}
  \boldsymbol{\mu}_i=({\boldsymbol{\lambda}^\mathcal{P}}'\boldsymbol{\lambda}^\mathcal{P})^{-1}{\boldsymbol{\lambda}^\mathcal{P}}'(\boldsymbol{Y}_i^\mathcal{P}-\boldsymbol{\delta}_i^\mathcal{P}-\boldsymbol{\varepsilon}_i^\mathcal{P}).
\end{align}

Substituting \eqref{Ch3_eq_mu} into $Y_{it,k}^0=\boldsymbol{X}_{it}'\boldsymbol{\beta}_{t,k}^0+\boldsymbol{\mu}_i'\boldsymbol{\lambda}_{t,k}^0+\varepsilon_{it,k}^0$, $t>T_0$, and with a little abuse on the notation by omitting the superscript $\mathcal{P}$ on the new coefficients and error term, we have
\begin{align}\label{Ch3_eq_IFE_rewrite}
  Y_{it,k}^0 & =\boldsymbol{X}_{it}'\boldsymbol{\beta}_{t,k}^0\underbrace{-\cdots-\boldsymbol{X}_{is}'\boldsymbol{\alpha}_{st,k}^0-\cdots}_{P\ \text{terms}}+{\boldsymbol{Y}_i^\mathcal{P}}'\boldsymbol{\gamma}_{t,k}^0+e_{it,k}^0,
\end{align}
where
\begin{align*}
  \boldsymbol{\alpha}_{st,k}^0 & ={\boldsymbol{\beta}_{s}^0}'\boldsymbol{\lambda}^0_{s}({\boldsymbol{\lambda}^\mathcal{P}}'\boldsymbol{\lambda}^\mathcal{P})^{-1}\boldsymbol{\lambda}^0_{t,k},\enspace s\in\mathcal{P}, \\
  \boldsymbol{\gamma}_{t,k}^0  & =\boldsymbol{\lambda}^\mathcal{P}({\boldsymbol{\lambda}^\mathcal{P}}'\boldsymbol{\lambda}^\mathcal{P})^{-1}\boldsymbol{\lambda}^0_{t,k},                                               \\
  e_{it,k}^0                   & =\varepsilon_{it,k}^0-{\boldsymbol{\gamma}_{t,k}^0}'\boldsymbol{\varepsilon}_i^\mathcal{P}.
\end{align*}

Let $Z=r(P+1)+KP$. If we denote the $Z\times 1$ vector of observables $[\boldsymbol{X}_{it}' \cdots \boldsymbol{X}_{is}' \cdots {\boldsymbol{Y}_i^\mathcal{P}}']'$ as $\boldsymbol{Z}_{it}$, and the $Z\times 1$ vector of coefficients $[{\boldsymbol{\beta}_{t,k}^0}' \cdots {\boldsymbol{\alpha}_{st,k}^0}' \cdots {\boldsymbol{\gamma}_{t,k}^0}']'$ as $\boldsymbol{\theta}_{t,k}^0$, then equation \eqref{Ch3_eq_IFE_rewrite} can be abbreviated as
\begin{equation}\label{Ch3_eq_IFE0_abb}
  Y_{it,k}^0=\boldsymbol{Z}_{it}'\boldsymbol{\theta}_{t,k}^0+e_{it,k}^0.
\end{equation}

Similarly, substituting \eqref{Ch3_eq_mu} into $Y_{it,k}^1=\boldsymbol{X}_{it}'\boldsymbol{\beta}_{t,k}^1+\boldsymbol{\mu}_i'\boldsymbol{\lambda}_{t,k}^1+\varepsilon_{it,k}^1$, $t>T_0$, we have
\begin{equation}\label{Ch3_eq_IFE1_abb}
  Y_{it,k}^1=\boldsymbol{Z}_{it}'\boldsymbol{\theta}_{t,k}^1+e_{it,k}^1,
\end{equation}
where
\begin{align*}
  \boldsymbol{\theta}_{t,k}^1  & =[{\boldsymbol{\beta}_{t,k}^1}' \cdots {\boldsymbol{\alpha}_{st,k}^1}' \cdots {\boldsymbol{\gamma}_{t,k}^1}']',                                                                        \\
  \boldsymbol{\alpha}_{st,k}^1 & ={\boldsymbol{\beta}_{s}^0}'\boldsymbol{\lambda}^0_{s}({\boldsymbol{\lambda}^\mathcal{P}}'\boldsymbol{\lambda}^\mathcal{P})^{-1}\boldsymbol{\lambda}_{t,k}^1,\enspace s\in\mathcal{P}, \\
  \boldsymbol{\gamma}_{t,k}^1  & =\boldsymbol{\lambda}^\mathcal{P}({\boldsymbol{\lambda}^\mathcal{P}}'\boldsymbol{\lambda}^\mathcal{P})^{-1}\boldsymbol{\lambda}_{t,k}^1,                                               \\
  e_{it,k}^1                   & =\varepsilon_{it,k}^1-{\boldsymbol{\gamma}_{t,k}^1}'\boldsymbol{\varepsilon}_i^\mathcal{P}.
\end{align*}

\subsection{Estimation}

Under Assumption \ref{Ch3_assume_error}, we have $\mathbb{E}(e_{it,k}^1\mid \boldsymbol{H}_{it})=0$ and $\mathbb{E}(e_{it,k}^0\mid \boldsymbol{H}_{it})=0$. This suggests using $\widehat{\tau}_{it,k}=\boldsymbol{Z}_{it}'(\widehat{\boldsymbol{\theta}}_{t,k}^1-\widehat{\boldsymbol{\theta}}_{t,k}^0)$, where $\widehat{\boldsymbol{\theta}}_{t,k}^1$ and $\widehat{\boldsymbol{\theta}}_{t,k}^0$ are some estimators of ${\boldsymbol{\theta}}_{t,k}^1$ and ${\boldsymbol{\theta}}_{t,k}^0$, to estimate $\bar{\tau}_{it,k}$.
Note, however, that the error terms $e_{it,k}^1$ and $e_{it,k}^0$ are correlated with the regressors, since $\boldsymbol{Z}_{it}$ contains $\boldsymbol{Y}_i^\mathcal{P}$ which is correlated with $\boldsymbol{\varepsilon}_i^\mathcal{P}$.
% $\boldsymbol{\varepsilon}_i^\mathcal{P}=\boldsymbol{Y}_i^\mathcal{P}-\boldsymbol{\delta}_i^\mathcal{P}-\boldsymbol{\lambda}^\mathcal{P}\boldsymbol{\mu}_i$ is a function of $\boldsymbol{Z}_{it}$.
This renders the OLS estimators biased and inconsistent, which can be seen as a classical measurement errors in variables problem.\footnote{This is noted in \cite{ferman2019synthetic} as well, who also suggested using pre-treatment outcomes as instrumental variables to deal with the problem. Our method is also related to the quasi-differencing approach in \cite{holtz1988estimating} and the GMM approach in \cite{ahn2013panel}. While these studies focus on estimating the coefficients on the observed covariates, our focus is on estimating the individual treatment effects.} We thus use the remaining outcomes as instrumental variables for $\boldsymbol{Y}_i^\mathcal{P}$ to consistently estimate ${\boldsymbol{\theta}}_{t,k}^1$ and ${\boldsymbol{\theta}}_{t,k}^0$ in each period, which would then allow us to obtain asymptotically unbiased estimates for the individual treatment effects.\footnote{We may construct the vectors of regressors and instruments differently under alternative assumptions on the dependence structure of the idiosyncratic shocks. For example, if the idiosyncratic shocks are correlated across time but are independent across outcomes, then we can split different outcomes into regressors and instruments. This would be similar to using the characteristics of similar products \citep{berry1995automobile} or trading countries (see the Trade-weighted World Income instrument in \citealp{acemoglu2008income}) as instrumental variables. Incorporating more complex structures of the idiosyncratic shocks in the model is left for future research.}
Since the outcomes depend on about the same set of observed and unobserved individual characteristics, the remaining outcomes are strongly correlated with the outcomes included in $\boldsymbol{Y}_i^\mathcal{P}$. Additionally, given that the idiosyncratic shocks are independent across time, the remaining outcomes are not correlated with $e_{it,k}^1$ or $e_{it,k}^0$. Thus, both the relevance and exogeneity conditions are satisfied, and the remaining outcomes can serve as valid instrumental variables.
% Given that the outcomes depend on about the same set of observed and unobserved individual characteristics and the assumptions on the idiosyncratic shocks, it follows that the relevance and exogeneity conditions for the instrumental variables are satisfied.

Let $\boldsymbol{R}_{it}=[\boldsymbol{X}_{it}' \cdots \boldsymbol{X}_{is}' \cdots {\boldsymbol{Y}_i^{-\mathcal{P}}}']'$ be the $R\times1$ vector of instruments, where the $(KT-KP-1)\times1$ vector $\boldsymbol{Y}_i^{-\mathcal{P}}$ comprises the remaining pretreatment outcomes as well as the posttreatment outcomes other than $Y_{it,k}$.\footnote{In the special case of $T_1=1$ and $T_0=1$, we can include $K-1$ pretreatment outcomes as regressors, and use the posttreatment outcomes other than $Y_{it,k}$ as instruments so that $R\ge Z$.}
Stacking $\boldsymbol{Z}_{it}$, $\boldsymbol{R}_{it}$ and $\boldsymbol{Y}_{it,k}^0$ respectively over the $N_0$ untreated individuals, we obtain the $N_0\times Z$ matrix of regressors $\boldsymbol{Z}_{t}^0$, the $N_0\times R$ matrix of instruments $\boldsymbol{R}_{t}^0$ and the $N_0\times 1$ matrix of outcomes $\boldsymbol{Y}_{t,k}^0$ for the untreated individuals. We can obtain $\boldsymbol{Z}_{t}^1$, $\boldsymbol{R}_{t}^1$ and $\boldsymbol{Y}_{t,k}^1$ similarly for the $N_1$ treated individuals.
The GMM estimator for the individual treatment effect $\bar{\tau}_{it,k}$ can then be constructed as
\begin{align}\label{Ch3_eq_tau}
  \widehat{\tau}_{it,k} & =\boldsymbol{Z}_{it}'\left(\widehat{\boldsymbol{\theta}}_{t,k}^1-\widehat{\boldsymbol{\theta}}_{t,k}^0\right),
\end{align}
where
\begin{align}
  \widehat{\boldsymbol{\theta}}_{t,k}^1 & =\left({\boldsymbol{Z}_{t}^1}'\boldsymbol{R}_{t}^1\boldsymbol{W}^1{\boldsymbol{R}_{t}^1}'\boldsymbol{Z}_{t}^1\right)^{-1}{\boldsymbol{Z}_{t}^1}'\boldsymbol{R}_{t}^1\boldsymbol{W}^1{\boldsymbol{R}_{t}^1}'\boldsymbol{Y}_{t,k}^1, \label{Ch3_eq_theta1} \\
  \widehat{\boldsymbol{\theta}}_{t,k}^0 & =\left({\boldsymbol{Z}_{t}^0}'\boldsymbol{R}_{t}^0\boldsymbol{W}^0{\boldsymbol{R}_{t}^0}'\boldsymbol{Z}_{t}^0\right)^{-1}{\boldsymbol{Z}_{t}^0}'\boldsymbol{R}_{t}^0\boldsymbol{W}^0{\boldsymbol{R}_{t}^0}'\boldsymbol{Y}_{t,k}^0, \label{Ch3_eq_theta0}
\end{align}
with $\boldsymbol{W}^1$ and $\boldsymbol{W}^0$ being some $R\times R$ positive definite matrices.

\begin{remark}\label{Ch3_remark_two_step}
  \normalfont
  Using the residuals $\widehat{\boldsymbol{e}}_{t,k}^1=\boldsymbol{Y}_{t,k}^1-\boldsymbol{Z}_{t}^1\widehat{\boldsymbol{\theta}}_{t,k}^1$ and $\widehat{\boldsymbol{e}}_{t,k}^0=\boldsymbol{Y}_{t,k}^0-\boldsymbol{Z}_{t}^0\widehat{\boldsymbol{\theta}}_{t,k}^0$, we can further construct the two-step efficient GMM estimator by replacing $\boldsymbol{W}^1$ and $\boldsymbol{W}^0$ in equations \eqref{Ch3_eq_theta1} and \eqref{Ch3_eq_theta0} with $N_1\left({\boldsymbol{R}_{t}^1}'\boldsymbol{U}_{t}^1\boldsymbol{R}_{t}^1\right)^{-1}$ and $N_0\left({\boldsymbol{R}_{t}^0}'\boldsymbol{U}_{t}^0\boldsymbol{R}_{t}^0\right)^{-1}$, where $\boldsymbol{U}_{t}^1$ and $\boldsymbol{U}_{t}^0$ are diagonal matrices with the squared elements of $\widehat{\boldsymbol{e}}_{t,k}^1$ and $\widehat{\boldsymbol{e}}_{t,k}^0$ on the diagonals.
\end{remark}

\begin{remark}
  \normalfont
  One may also construct the estimators for the individual treatment effects using authentic predicted outcomes obtained from a leave-one-out procedure, where ${\boldsymbol{\theta}}_{t,k}^1$ and ${\boldsymbol{\theta}}_{t,k}^0$ are estimated for each individual using the sample that excludes that individual. This procedure may be computationally expensive though, as there are no simple linear expressions for the leave-one-out coefficients estimates and residuals as for those in linear regression \citep{hansen2021econometrics}.
\end{remark}

\medskip

The following result shows that the bias of the GMM estimator for the individual treatment effect in \eqref{Ch3_eq_tau} goes away as both the number of treated individuals and the number of untreated individuals become larger.

\begin{prop}\label{Ch3_prop_ITE_GMM}
  Under Assumptions \ref{Ch3_assume_predictor}-\ref{Ch3_assume_rank},
  $\mathbb{E}\left(\widehat{\tau}_{it,k}-\tau_{it,k}\mid \boldsymbol{H}_{it}=\boldsymbol{h}_{it}\right)\rightarrow0$ as $N_1,N_0\rightarrow\infty$.
\end{prop}
% result means that in large samples, the estimated ITEs center around the true ITEs.

\medskip

Once we have the estimates for the individual treatment effects, the average treatment effect $\tau_{t,k}=\mathbb{E}\left(\tau_{it,k}\right)$ can be conveniently estimated using the average of the estimated individual treatment effects $\widehat{\tau}_{t,k}=\frac{1}{N}\sum_{i=1}^N\widehat{\tau}_{it,k}$, which can be shown to be consistent.

\begin{prop}\label{Ch3_prop_ATE_GMM}
  Under Assumptions \ref{Ch3_assume_predictor}-\ref{Ch3_assume_rank},
  $\widehat{\tau}_{t,k}-\tau_{t,k}\overset{p}{\rightarrow}0$ as $N_0,N_1\rightarrow\infty$, and $\widehat{\tau}_{t,k}-\tau_{t,k}=O_p\left(N_1^{-1/2}\right)+O_p\left(N_0^{-1/2}\right)$.
\end{prop}

\subsection{Model Selection}\label{Ch3_sec_ms}

To satisfy Assumption \ref{Ch3_assume_rank}, we need the number of pretreatment outcomes that we include as regressors in the model to be at least as large as $f$.
Including more pretreatment outcomes may increase the variance of the estimator by increasing the variances of $\widehat{\boldsymbol{\theta}}_{t,k}^1$ and $\widehat{\boldsymbol{\theta}}_{t,k}^0$, but may also reduce the variance of the estimator when the sample is large and the variances of $\widehat{\boldsymbol{\theta}}_{t,k}^1$ and $\widehat{\boldsymbol{\theta}}_{t,k}^0$ are small, since
\begin{align}\label{Ch3_eq_error}
  \left(\boldsymbol{\gamma}_{t,k}^1-\boldsymbol{\gamma}_{t,k}^0\right)'\boldsymbol{\varepsilon}_i^\mathcal{P}=\frac{1}{KP}\sum_{q\in\mathcal{K}}\sum_{s\in\mathcal{P}}\left(\boldsymbol{\lambda}^1_{t,k}-\boldsymbol{\lambda}^0_{t,k}\right)'\left(\frac{1}{KP}\sum_{l\in\mathcal{K}}\sum_{n\in\mathcal{P}}\boldsymbol{\lambda}_{n,l}^0{\boldsymbol{\lambda}_{n,l}^0}'\right)^{-1}\boldsymbol{\lambda}_{s,q}^0\varepsilon_{is,q}^0
\end{align}
in the prediction error converges in probability to 0 as $KP$ grows.\footnote{Consistency of the individual treatment effect estimator may also be shown by allowing both $N$ and $KP$ to grow, with restrictions on the relative growth rate, e.g., $\frac{KP}{\min\left(\sqrt{N_1},\sqrt{N_0}\right)}\rightarrow0$. We do not pursue this path in this study, as the number of pretreatment outcomes in empirical microeconomics that we focus on is usually not large.}

To select the number of pretreatment outcomes to include in the model, we follow a model selection procedure similar to that in \cite{hsiao2012panel}, where for each usable number of pretreatment outcomes, we construct many different models by including a random subset of the pretreatment outcomes as regressors and the remaining outcomes as instruments. We then estimate the models using GMM and obtain the leave-one-out prediction errors for all or a subsample of the individuals.
The best set of pretreatment outcomes is chosen as the one that minimises the mean squared leave-one-out prediction error.\footnote{An alternative way to select the best set of pretreatment outcomes is to use information criteria such as GMM-BIC and GMM-AIC \citep{andrews1999consistent}. To avoid the potential problem of post-selection inference, we may also randomly split the sample into two parts, where we select the best model on one part, and conduct inference on the other.}

In addition to the models using only a subset of the pretreatment outcomes, we also consider averaging different models that use the same number of pretreatment outcomes. Since the estimators constructed using only a subset of the pretreatment outcomes are asymptotically unbiased, as long as the number of pretreatment outcomes is larger than $f$, this property is passed on to the averaged estimator. The averaged estimator may also be more efficient as it uses more information in the sample and reduces uncertainty caused by a small number of sample splits.\footnote{We stick with simple averaging in this paper. More flexible averaging scheme, e.g., with larger weights on those with smaller out of sample prediction errors, would be an interesting direction for future research.}
The leave-one-out prediction errors are also averaged over the models, and the best number of pretreatment outcomes to be used for the averaged estimator is similarly determined by minimising the mean squared leave-one-out prediction error.

\subsection{Related methods}

\subsubsection{Linear conditional mean}

An alternative approach to estimating the treatment effects is to follow \cite{hsiao2012panel} and assume that
\begin{align}
  \mathbb{E}\left(\boldsymbol{\varepsilon}_i^\mathcal{P}\mid\boldsymbol{Z}_{it}\right)=\boldsymbol{C}'\boldsymbol{Z}_{it},
\end{align}
where $\boldsymbol{C}=\mathbb{E}\left(\boldsymbol{Z}_{it}\boldsymbol{Z}_{it}'\right)^{-1}\mathbb{E}\left(\boldsymbol{Z}_{it}{\boldsymbol{\varepsilon}_i^\mathcal{P}}'\right)$ is $Z\times KT_0$.\footnote{This assumption holds in special cases, e.g., when the unobserved predictors and the idiosyncratic shocks all follow the normal distribution \citep{li2017estimation}. In more general cases, this assumption may be considered to hold approximately.}
We can then separate the error term into a part correlated with the regressors and a part that has zero conditional mean, and rewrite the untreated potential outcome $Y_{it,k}^0$ as
\begin{align*}
  Y_{it,k}^0 & =\mathbb{E}\left(Y_{it,k}^0\mid \boldsymbol{Z}_{it}\right)+u_{it,k}^0   \nonumber                                        \\
             & =\left({\boldsymbol{\theta}_{t,k}^0}'-{\boldsymbol{\gamma}_{t,k}^0}'\boldsymbol{C}'\right)\boldsymbol{Z}_{it}+u_{it,k}^0 \\
             & =\boldsymbol{Z}_{it}'\boldsymbol{\theta}_{t,k}^{*0}+u_{it,k}^0, \numberthis
\end{align*}
where $u_{it,k}^0=\varepsilon_{it,k}^0-{\boldsymbol{\gamma}_{t,k}^0}'\boldsymbol{\varepsilon}_i^\mathcal{P}+{\boldsymbol{\gamma}_{t,k}^0}'\boldsymbol{C}'\boldsymbol{Z}_{it}$.
Similarly, the treated potential outcome $Y_{it,k}^1$ can be rewritten as
\begin{align}
  Y_{it,k}^1 & =\boldsymbol{Z}_{it}'\boldsymbol{\theta}_{t,k}^{*1}+u_{it,k}^1,
\end{align}
where $\boldsymbol{\theta}_{t,k}^{*1}={\boldsymbol{\theta}_{t,k}^1}-\boldsymbol{C}\boldsymbol{\gamma}_{t,k}^1$, and $u_{it,k}^1=\varepsilon_{it,k}^1-{\boldsymbol{\gamma}_{t,k}^1}'\boldsymbol{\varepsilon}_i^\mathcal{P}+{\boldsymbol{\gamma}_{t,k}^1}'\boldsymbol{C}'\boldsymbol{Z}_{it}$.

Since $\mathbb{E}(u_{it,k}^1\mid\boldsymbol{Z}_{it})=\mathbb{E}[e_{it,k}^1-\mathbb{E}(e_{it,k}^1\mid\boldsymbol{Z}_{it})\mid \boldsymbol{Z}_{it}]=0$ and $\mathbb{E}(u_{it,k}^0\mid\boldsymbol{Z}_{it})=0$, it is straightforward to show that the least squares estimators $\widehat{\boldsymbol{\theta}}_{t,k}^{*1}=({\boldsymbol{Z}_{t}^1}'\boldsymbol{Z}_{t}^1)^{-1}{\boldsymbol{Z}_{t}^1}'\boldsymbol{Y}_{t,k}^1$ and $\widehat{\boldsymbol{\theta}}_{t,k}^{*0}=({\boldsymbol{Z}_{t}^0}'\boldsymbol{Z}_{t}^0)^{-1}{\boldsymbol{Z}_{t}^0}'\boldsymbol{Y}_{t,k}^0$ are the unbiased estimators of $\boldsymbol{\theta}_{t,k}^{*1}$ and $\boldsymbol{\theta}_{t,k}^{*0}$ respectively.\footnote{The linear conditional mean assumption also implies that the unconfoundedness assumption is satisfied, as $\mathbb{E}(Y_{it,k}^0\mid\boldsymbol{Z}_{it},D_{it}=1)=\mathbb{E}(Y_{it,k}^0\mid\boldsymbol{Z}_{it},D_{it}=0)$ and $\mathbb{E}(Y_{it,k}^1\mid\boldsymbol{Z}_{it},D_{it}=1)=\mathbb{E}(Y_{it,k}^1\mid\boldsymbol{Z}_{it},D_{it}=0)$.}

We can then construct an estimator as
\begin{align}
  \tilde{\tau}_{it,k} & =\boldsymbol{Z}_{it}'\left(\widehat{\boldsymbol{\theta}}_{t,k}^{*1}-\widehat{\boldsymbol{\theta}}_{t,k}^{*0}\right),
\end{align}
which is an unbiased estimator for the average treatment effect for individuals with the same values of $\boldsymbol{Z}_{it}$, or the conditional average treatment effect.
It follows that the average of the conditional average treatment effects estimators $\tilde{\tau}_{t,k}=\frac{1}{N}\sum_{i=1}^N\tilde{\tau}_{it,k}$ is an unbiased estimator for the average treatment effect $\tau_{t,k}$.
In addition, it can also be shown that $\tilde{\tau}_{t,k}$ is a consistent estimator without imposing the linear conditional mean assumption \citep{li2017estimation}.

\begin{prop}\label{Ch3_prop_OLS}
  Under Assumptions \ref{Ch3_assume_predictor}-\ref{Ch3_assume_rank},
  \begin{enumerate}[label=(\roman*)]
    \item if $\mathbb{E}\left(\boldsymbol{\varepsilon}_i^\mathcal{P}\mid\boldsymbol{Z}_{it}\right)=\boldsymbol{C}'\boldsymbol{Z}_{it}$, then $\mathbb{E}\left(\tilde{\tau}_{it,k}-\tau_{it,k}\mid \boldsymbol{Z}_{it}=\boldsymbol{z}_{it}\right)=0$ and $\mathbb{E}\left(\tilde{\tau}_{t,k}-\tau_{t,k}\right)=0$;
    \item $\tilde{\tau}_{t,k}-\tau_{t,k}=O_p\left(N_1^{-1/2}\right)+O_p\left(N_0^{-1/2}\right)$.
  \end{enumerate}
\end{prop}

\begin{remark}
  \normalfont
  Note that $\mathbb{E}\left(\tau_{it,k}\mid \boldsymbol{Z}_{it}=\boldsymbol{z}_{it}\right)$ is the average treatment effect for individuals with $\boldsymbol{Z}_{it}=\boldsymbol{z}_{it}$, or the conditional average treatment effect, whereas the individual treatment effect is $\mathbb{E}\left(\tau_{it,k}\mid \boldsymbol{H}_{it}=\boldsymbol{h}_{it}\right)$ as given in \eqref{Ch3_eq_ITE1}. The two are generally not the same since $\boldsymbol{C}'\boldsymbol{Z}_{it}\neq 0$.
\end{remark}

\subsubsection{Interactive fixed effects model}

Instead of replacing the unobserved confounders with the observed pretreatment outcomes, \cite{bai2009panel} models the unobserved fixed effects directly by iterating between estimating the coefficients on the observed covariates and estimating the unobserved factors and factor loadings using the principal component analysis, given some initial values.
This approach allows more general structures in the error terms, but requires both $N$ and $T$ to be large, and is also more restrictive on the model specification: the observed covariates need to be time-varying, while the coefficients are assumed constant over time.
\cite{xu2017generalized} adapts this method to the potential outcomes framework to estimate the average treatment effects on the treated, assuming that the untreated potential outcomes for both the treated and untreated units follow the interactive fixed effects model, and proposes a cross-validation procedure to choose the number of unobserved factors and a parametric bootstrap procedure for inference.

This approach has the desired feature of being less computationally expensive compared with repeated pretreatment set splitting and averaging, and is potentially more efficient compared with using only the best set of pretreatment outcomes and discarding the remaining information when all outcomes are related. However, its potential to be adapted to our settings is limited by the restrictions discussed above. In particular, we may assume the coefficients to be constant over time, but it would be unrealistic to assume that they are the same across different outcomes, if we were to use multiple related outcomes.

Another closely related study is \cite{athey2021matrix}, which generalises the results from the matrix completion literature in computer science to impute the missing elements of the untreated potential outcome matrix for the treated units in the posttreatment periods, where the matrix is assumed to have a low rank structure, similar to that of the interactive fixed effects model. The bias of the estimator is shown to have an upper bound that goes to 0 as both $N$ and $T$ grow. This method allows staggered adoption of the treatment, i.e., the treated units receive the treatment at different time periods.

\subsubsection{Synthetic control method}

\cite{abadie2010synthetic} estimate the treatment effect on a treated unit by predicting its untreated potential outcome using a synthetic control constructed as a weighted average of the control units.
The synthetic control method applies to cases where the pretreatment characteristics of the treated unit can be closely approximated by the synthetic control constructed using a small number of control units over an extended period of time before the treatment, which may not generally hold.
In terms of implementation, the objective function for the synthetic control method is similar to that of the linear regression approach in \cite{hsiao2012panel}. However, the weights on the control units in the synthetic control method are restricted to be nonnegative to avoid extrapolation. This reduces the risk of overfitting, but may also limit its applicability by making it difficult to find a set of weights that satisfy the restrictions.

\subsection{Inference}\label{Ch3_sec_ci}

\sloppy To measure the conditional variance of the individual treatment effect estimator, $\text{Var}\left(\widehat{\tau}_{it,k}\mid \boldsymbol{H},\boldsymbol{D}\right)$, where $\boldsymbol{H}$ is the matrix of observed covariates and unobserved individual characteristics and $\boldsymbol{D}$ is the matrix of the treatment status for all individuals and all time periods in the sample, we follow \cite{xu2017generalized} and employ a parametric bootstrap procedure.
% by fixing the fitted values of the outcomes and resampling the residuals.

First, we apply our method to all outcomes in all periods to obtain $\widehat{Y}_{it,k}^1$ and $\widehat{e}_{it,k}^1$ for the treated individuals in the posttreatment periods, and $\widehat{Y}_{it,k}^0$ and $\widehat{e}_{it,k}^0$ for the untreated individuals in the posttreatment periods and for all individuals in the pretreatment periods.
Note that the residuals $\widehat{e}_{it,k}^1$ and $\widehat{e}_{it,k}^0$ are estimates for $\varepsilon_{it,k}^1-{\boldsymbol{\gamma}_{t,k}^1}'\boldsymbol{\varepsilon}_i^\mathcal{P}$ and $\varepsilon_{it,k}^0-{\boldsymbol{\gamma}_{t,k}^0}'\boldsymbol{\varepsilon}_i^\mathcal{P}$, respectively, rather than the idiosyncratic shocks in the original model, $\varepsilon_{it,k}^1$ and $\varepsilon_{it,k}^0$. Thus, the variance of the individual treatment effect estimator tends to be overestimated using the parametric bootstrap by resampling these residuals, especially when the number of pretreatment outcomes is small.\footnote{See the discussion on equation \eqref{Ch3_eq_error}.} Correcting for this bias would be a necessary step for future research.

These fitted values of the outcomes can be stacked into a $TK\times 1$ vector $\widehat{\boldsymbol{Y}}_i$ for each individual, where $\widehat{\boldsymbol{Y}}_i$ for $i\in\mathcal{T}$ contains $\widehat{Y}_{it,k}^1$ in the posttreatment periods and $\widehat{Y}_{it,k}^0$ in the pretreatment periods, and $\widehat{\boldsymbol{Y}}_i$ for $i\in\mathcal{C}$ contains $\widehat{Y}_{it,k}^0$ in all periods.
The $TK\times 1$ vector of residuals $\widehat{\boldsymbol{e}}_i$ can be obtained similarly.

We then start bootstrapping for $B$ rounds:
\begin{enumerate}
  \item In round $b\in\{1,\dots,B\}$, generate a bootstrapped sample as
        \begin{align*}
          \boldsymbol{Y}_i^{(b)} & = \widehat{\boldsymbol{Y}}_i + \widehat{\boldsymbol{e}}_i^{(b)}, \enspace \text{for all} \enspace i,
        \end{align*}
        where $\widehat{\boldsymbol{e}}_i^{(b)}$ is randomly drawn from $\{\widehat{\boldsymbol{e}}_i\}_{i\in\mathcal{T}}$ for $i\in\mathcal{T}$, and from $\{\widehat{\boldsymbol{e}}_i\}_{i\in\mathcal{C}}$ for $i\in\mathcal{C}$.\footnote{Since the entire series of residuals over the $T$ periods and $K$ outcomes are resampled, correlation and heteroskedasticity across time and outcomes are preserved \citep{xu2017generalized}.}
        % , as the variances of the estimated idiosyncratic shocks may be different for the treated  and untreated groups due to different group sizes.
  \item Construct $\widehat{\tau}_{it,k}^{(b)}$ for each $i$ using the above bootstrapped sample.
\end{enumerate}

The variance for the individual treatment effect estimator is computed using the bootstrap estimates as
\begin{align*}
  \text{Var}\left(\widehat{\tau}_{it,k}\mid \boldsymbol{H},\boldsymbol{D}\right)=\frac{1}{B}\sum_{b=1}^B\left(\widehat{\tau}_{it,k}^{(b)}-\frac{1}{B}\sum_{a=1}^B\widehat{\tau}_{it,k}^{(a)}\right),\enspace i=1,\dots,N,
\end{align*}
and the $100(1-\alpha)\%$ confidence intervals for $\bar{\tau}_{it,k}$, $i=1,\dots,N$ can be constructed as
\begin{equation*}
  \left[ \widehat{\tau}_{it,k}^{\left[\frac{\alpha}{2}B\right]}\,, \;\widehat{\tau}_{it,k}^{\left[(1-\frac{\alpha}{2})B\right]} \right],
\end{equation*}
where the superscript denotes the index of the bootstrap estimates in ascending order.
Alternatively, we can use a normal approximation and construct the confidence intervals as
\begin{equation*}
  \left[ \widehat{\tau}_{it,k}+\Phi^{-1}\left(\frac{\alpha}{2}\right)\widehat{\sigma}_{it,k}\,, \;\widehat{\tau}_{it,k}+\Phi^{-1}\left(1-\frac{\alpha}{2}\right)\widehat{\sigma}_{it,k} \right],
\end{equation*}
where $\Phi\left(\cdot\right)$ is the cumulative distribution function for the standard normal distribution, and $\widehat{\sigma}_{it,k}=\sqrt{\text{Var}\left(\widehat{\tau}_{it,k}\mid \boldsymbol{H},\boldsymbol{D}\right)}$.

The variance for the average treatment effect estimator $\widehat{\tau}_{t,k}=\frac{1}{N}\sum_{i=1}^N\widehat{\tau}_{it,k}$ and the confidence interval for the average treatment effect $\tau_{t,k}$ can be obtained in similar manners using the bootstrap estimates $\widehat{\tau}_{t,k}^{(b)}=\frac{1}{N}\sum_{i=1}^N\widehat{\tau}_{it,k}^{(b)}$, $b=1,\dots,B$.

\section{Monte Carlo Simulations}\label{Ch3_sec_sim}

In this section, we conduct Monte Carlo simulations to assess the performance of our estimator in small samples, and compare it with related methods in relevant settings.
The number of posttreatment period $T_1$ is fixed at 1, and the number of related outcomes $K$ is fixed at 5 in all settings.

The untreated potential outcomes are generated from
\begin{equation}
  Y_{it,k}^0=\boldsymbol{X}_{it}'\boldsymbol{\beta}_{t,k}^0+\boldsymbol{\mu}_i'\boldsymbol{\lambda}_{t,k}^0+\varepsilon_{it,k}^0,\enspace k\in\mathcal{K},
\end{equation}
where $\boldsymbol{X}_{it}$ contains 2 observed covariates, and $\boldsymbol{\mu}_{i}$ contains 2 unobserved individual characteristics as well as the constant 1.
The 2 observed covariates are i.i.d. $N(0,1)$ in period 1, and then follow an AR(1) process, $\boldsymbol{X}_{it}=0.9\boldsymbol{X}_{i,t-1}+\xi_{it}$, where $\xi_{it}$ are i.i.d. $N(0,\sqrt{1-0.9^2})$, so that the observed covariates are correlated across time and the variances stay 1.
The 2 unobserved individual characteristics are also i.i.d. $N(0,1)$.
The coefficients $\boldsymbol{\beta}_{t,k}^0$ and $\boldsymbol{\lambda}_{t,k}^0$ are i.i.d. $N(\omega_k,1)$ with $\omega_k\sim N(1,1)$, for $k\in\mathcal{K}$, so that the means of the coefficients differ across outcomes, and the idiosyncratic shocks $\varepsilon_{it,k}^0$ are i.i.d. $N(0,1)$.

The individual treatment effect in the posttreatment period $\bar{\tau}_{iT_0+1,k}$ is a deterministic function of $\boldsymbol{X}_{it}$ and $\boldsymbol{\mu}_i$ with the coefficients being i.i.d. $N(0.5,0.5)$, for $k\in\mathcal{K}$.
And the observed outcomes $Y_{it,k}$, $k\in\mathcal{K}$ are equal to $Y_{it,k}^0-\varepsilon_{it,k}^0+\bar{\tau}_{iT_0+1,k}+\varepsilon_{it,k}^1$, where $\varepsilon_{it,k}^1$ are i.i.d. $N(0,1)$, for the treated individuals in the posttreatment period, and $Y_{it,k}^0$ otherwise.
$\boldsymbol{X}_{it}$ and $\boldsymbol{\mu}_i$ as well as their coefficients for the untreated potential outcomes and the treatment effects are drawn 5 times, and for each set of $\{\boldsymbol{X}_{it},\boldsymbol{\mu}_i\}$ and their coefficients drawn, $\varepsilon_{it,k}^0$ and $\varepsilon_{it,k}^1$ are drawn 1000 times, which allows us to compute the bias and variance of the estimator conditional on the observed covariates and the unobserved individual characteristics.

\sloppy To measure the performances of the estimators, we compute the biases and standard deviations for the estimates of the individual treatment effects and the average treatment effect for outcome $K$ in the posttreatment period.
Specifically, the bias of the individual treatment effect estimator $\widehat{\tau}_{iT_0+1,K}$ is measured by $\frac{1}{N}\sum_{i=1}^N\frac{1}{5}\sum_{d=1}^5\left\vert\mathbb{E}\left(\widehat{\tau}_{iT_0+1,K}^{(d,s)}\right)-\bar{\tau}_{iT_0+1,K}^{(d)}\right\vert$, where the superscript $d$ denotes the $d$th draw of $\{\boldsymbol{X}_{it},\boldsymbol{\mu}_i\}$ and $s$ denotes the $s$th draw of $\varepsilon_{it,k}^0$ and $\varepsilon_{it,k}^1$, and the standard deviation is constructed as $\frac{1}{N}\sum_{i=1}^N\frac{1}{5}\sum_{d=1}^5\sqrt{\mathbb{E}\left(\widehat{\tau}_{iT_0+1,K}^{(d,s)}-\mathbb{E}\widehat{\tau}_{iT_0+1,K}^{(d,s)}\right)^2}$.
Similarly, the bias of the average treatment effect estimator $\widehat{\tau}_{T_0+1,K}$ is measured by $\frac{1}{5}\sum_{d=1}^5\left\vert\mathbb{E}\left(\widehat{\tau}_{T_0+1,K}^{(d,s)}\right)-\bar{\tau}_{T_0+1,K}^{(d)}\right\vert$, and the standard deviation is constructed as $\frac{1}{5}\sum_{d=1}^5\sqrt{\mathbb{E}\left(\widehat{\tau}_{T_0+1,K}^{(d,s)}-\mathbb{E}\widehat{\tau}_{T_0+1,K}^{(d,s)}\right)^2}$.\footnote{The performance of the estimators can also be measured using RMSE, which is computed as $\frac{1}{5}\sum_{d=1}^5\sqrt{\mathbb{E}\left(\widehat{\tau}_{iT_0+1,K}^{(d,s)}-\mathbb{E}{\bar{\tau}}_{iT_0+1,K}^{(d,s)}\right)^2}$ for $\widehat{\tau}_{iT_0+1,K}$ and $\frac{1}{5}\sum_{d=1}^5\sqrt{\mathbb{E}\left(\widehat{\tau}_{T_0+1,K}^{(d,s)}-\mathbb{E}{\bar{\tau}}_{T_0+1,K}^{(d,s)}\right)^2}$ for $\widehat{\tau}_{T_0+1,K}$. Since the biases of our estimators are small, these measures are quite similar to SD and are thus omitted from reporting.}

\begin{center}
  \resizebox{13cm}{!}{
    \begin{threeparttable}[!htbp]
      \centering
      \caption{Simulation Results on Model Selection}\label{Ch3_tab_sim1}
      \begin{tabular}{ccc@{\hskip 0.7ex}cccccc @{\hskip 0.5ex}cccccc }
        \toprule
        \multicolumn{4}{c}{} & \multicolumn{5}{c}{Best Set} & \multicolumn{1}{c}{}    & \multicolumn{5}{c}{Model Averaging}                                                                                                                      \\
        \cmidrule(lr){5-9} \cmidrule(lr){11-15}
        \multicolumn{5}{c}{} & \multicolumn{2}{c}{ITE}      & \multicolumn{2}{c}{ATE} & \multicolumn{2}{c}{}                & \multicolumn{2}{c}{ITE} & \multicolumn{2}{c}{ATE}                                                                  \\
        \cmidrule(lr){6-7} \cmidrule(lr){8-9}\cmidrule(lr){12-13} \cmidrule(lr){14-15}
        $N_1$                & $N_0$                        & $T_0$                   &                                     & P                       & Bias                    & SD    & Bias  & SD    &  & P   & Bias  & SD    & Bias  & SD    \\
        \midrule
        50                   & 50                           & 1                       &                                     & 2.2                     & 0.151                   & 1.225 & 0.082 & 0.384 &  & 2.3 & 0.231 & 1.163 & 0.120 & 0.336 \\
        100                  & 100                          & 1                       &                                     & 2.2                     & 0.076                   & 0.764 & 0.032 & 0.212 &  & 2.4 & 0.065 & 0.836 & 0.024 & 0.205 \\
        200                  & 200                          & 1                       &                                     & 2.1                     & 0.038                   & 0.476 & 0.004 & 0.127 &  & 2.6 & 0.046 & 0.712 & 0.011 & 0.133 \\
        \addlinespace[1ex]
        50                   & 50                           & 2                       &                                     & 2.6                     & 0.062                   & 0.875 & 0.014 & 0.253 &  & 2.9 & 0.150 & 0.758 & 0.040 & 0.232 \\
        100                  & 100                          & 2                       &                                     & 2.7                     & 0.035                   & 0.685 & 0.003 & 0.165 &  & 2.9 & 0.073 & 0.563 & 0.014 & 0.159 \\
        200                  & 200                          & 2                       &                                     & 3.2                     & 0.038                   & 0.729 & 0.003 & 0.137 &  & 4.0 & 0.031 & 0.702 & 0.003 & 0.131 \\
        \bottomrule
      \end{tabular}
      \begin{tablenotes}
        \item Note: This table compares the estimator using only the best set of pretreatment outcomes and the estimator constructed from model averaging, in terms of the optimal number of pretreatment outcomes selected by LOO cross-validation, as well as the bias and SD for the ITE and ATE estimates, with varying sample size and number of pretreatment periods, based on 5000 simulations for each setting.
      \end{tablenotes}
    \end{threeparttable}
  }
\end{center}

\medskip

Table \ref{Ch3_tab_sim1} compares the GMM estimator constructed using only the best set of pretreatment outcomes with that constructed by averaging estimators from different models with the same number of pretreatment outcomes.
We see that the best number of pretreatment outcomes, $P$, is slightly larger than the number of unobserved individual characteristics ($f=2$) for both estimators, and increases when the sample size is larger and when there are more pretreatment outcomes available, which is in line with our discussions in section \ref{Ch3_sec_ms}. The estimators constructed by model averaging also tends to select a slightly larger $P$ than the estimator using only the best set of pretreatment outcomes.

In almost all settings, the estimator using only the best set of pretreatment outcomes tends to have a smaller bias, whereas the estimator constructed from model averaging tends to have a smaller variance, except for estimating the individual treatment effects when the number of pretreatment outcomes is very small. The bias and SD also become smaller for both estimators when the sample size as well as the number of pretreatment outcomes grow.

In the following simulations, we fix $P$ at 2 when $T_0=1$, and 3 when $T_0=2$, and construct the GMM estimator using only the best set of pretreatment outcomes, with the best set of pretreatment outcomes selected at the first simulation and used for the remaining simulations for each setting.\footnote{This is mainly to save computing time and does not fundamentally change the conclusions.}

\begin{center}
  \resizebox{11cm}{!}{
    \begin{threeparttable}[!htbp]
      \centering
      \caption{Simulation Results for the GMM estimator}\label{Ch3_tab_sim2}
      \begin{tabular}{ccc@{\hskip 0.7ex}cccc @{\hskip 0.5ex}cccc }
        \toprule
        \multicolumn{4}{c}{} & \multicolumn{3}{c}{ITE} & \multicolumn{1}{c}{} & \multicolumn{3}{c}{ATE}                                                          \\
        \cmidrule(lr){5-7} \cmidrule(lr){9-11}
        $N_1$                & $N_0$                   & $T_0$                &                         & Bias  & SD    & Coverage &  & Bias  & SD    & Coverage \\
        \midrule
        \addlinespace[2ex]
        \multicolumn{11}{c}{Panel A: $\varepsilon_{it,k}^0$ uncorrelated across $t$ and $k$}                                                                     \\
        \addlinespace[1ex]
        50                   & 50                      & 1                    &                         & 0.096 & 1.422 & 0.997    &  & 0.040 & 0.443 & 0.995    \\
        100                  & 100                     & 1                    &                         & 0.043 & 0.856 & 0.992    &  & 0.005 & 0.226 & 0.984    \\
        50                   & 50                      & 2                    &                         & 0.043 & 1.163 & 0.973    &  & 0.011 & 0.288 & 0.959    \\
        100                  & 100                     & 2                    &                         & 0.025 & 0.900 & 0.953    &  & 0.005 & 0.181 & 0.956    \\
        \addlinespace[1ex]
        \multicolumn{11}{c}{Panel B: $\varepsilon_{it,k}^0$ correlated across $t$ and $k$}                                                                       \\
        \addlinespace[1ex]
        50                   & 50                      & 1                    &                         & 0.134 & 1.419 & 0.996    &  & 0.045 & 0.431 & 0.993    \\
        100                  & 100                     & 1                    &                         & 0.065 & 0.851 & 0.991    &  & 0.018 & 0.230 & 0.982    \\
        50                   & 50                      & 2                    &                         & 0.063 & 1.162 & 0.976    &  & 0.008 & 0.294 & 0.961    \\
        100                  & 100                     & 2                    &                         & 0.037 & 0.906 & 0.964    &  & 0.009 & 0.183 & 0.967    \\
        \bottomrule
      \end{tabular}
      \begin{tablenotes}
        \item Note: This table compares the bias and SD of the GMM estimator, as well as the coverage probability of the 95\% confidence interval, with varying sample size and number of pretreatment periods, based on 5000 simulations for each setting.
      \end{tablenotes}
    \end{threeparttable}
  }
\end{center}

\medskip

Table \ref{Ch3_tab_sim2} reports the bias and SD of the GMM estimator, as well as the coverage probability of the 95\% confidence interval, for estimating the individual treatment effects and the average treatment effect.
Panel A shows that the bias and SD for the estimators are small even with a small sample size and a small number of pretreatment outcomes. However, the 95\% confidence intervals tend to have larger coverage probabilities, especially when the number of pretreatment outcomes is small. This distortion is alleviated as more pretreatment outcomes are available.

Since the validity of the GMM estimator relies on the assumption that $\varepsilon_{it,k}$ are uncorrelated across time or outcomes, we examine the performance of the estimator when this assumption is violated in Panel B, where the idiosyncratic shocks follow an AR(1) process over time with the autoregression coefficient being 0.1, and are correlated across outcomes by sharing a common component for different outcomes in the same period. This slightly increases the biases and SD's of the estimators, but the performance of the estimators are still quite good, especially in comparison with related methods as shown in the following tables.

\begin{center}
  \resizebox{11cm}{!}{
    \begin{threeparttable}[!htbp]
      \centering
      \caption{Simulation}\label{Ch3_tab_sim3}
      \begin{tabular}{ccc@{\hskip 0.7ex}ccccc @{\hskip 0.5ex}ccccc }
        \toprule
        \multicolumn{4}{c}{} & \multicolumn{4}{c}{OLS} & \multicolumn{1}{c}{}    & \multicolumn{4}{c}{GMM}                                                                                                        \\
        \cmidrule(lr){5-8} \cmidrule(lr){10-13}
        \multicolumn{4}{c}{} & \multicolumn{2}{c}{ITE} & \multicolumn{2}{c}{ATE} & \multicolumn{1}{c}{}    & \multicolumn{2}{c}{ITE} & \multicolumn{2}{c}{ATE}                                                    \\
        \cmidrule(lr){5-6} \cmidrule(lr){7-8}\cmidrule(lr){10-11} \cmidrule(lr){12-13}
        $N_1$                & $N_0$                   & $T_0$                   &                         & Bias                    & SD                      & Bias  & SD    &  & Bias  & SD    & Bias  & SD    \\
        \midrule
        \addlinespace[2ex]
        \multicolumn{13}{c}{Panel A: Linear conditional mean}                                                                                                                                                     \\
        \addlinespace[1ex]
        100                  & 100                     & 1                       &                         & 0.099                   & 0.622                   & 0.041 & 0.186 &  & 0.113 & 1.297 & 0.086 & 0.257 \\
        200                  & 200                     & 1                       &                         & 0.059                   & 0.415                   & 0.015 & 0.119 &  & 0.011 & 0.430 & 0.003 & 0.129 \\
        100                  & 100                     & 2                       &                         & 0.046                   & 0.677                   & 0.012 & 0.158 &  & 0.026 & 0.901 & 0.003 & 0.179 \\
        200                  & 200                     & 2                       &                         & 0.064                   & 0.547                   & 0.012 & 0.121 &  & 0.028 & 0.937 & 0.003 & 0.142 \\
        \addlinespace[1ex]
        \multicolumn{13}{c}{Panel B: Nonlinear conditional mean}                                                                                                                                                  \\
        \addlinespace[1ex]
        100                  & 100                     & 1                       &                         & 0.097                   & 0.687                   & 0.057 & 0.199 &  & 0.025 & 0.806 & 0.011 & 0.244 \\
        200                  & 200                     & 1                       &                         & 0.140                   & 0.456                   & 0.040 & 0.122 &  & 0.016 & 0.552 & 0.003 & 0.149 \\
        100                  & 100                     & 2                       &                         & 0.077                   & 0.734                   & 0.015 & 0.173 &  & 0.115 & 1.118 & 0.007 & 0.213 \\
        200                  & 200                     & 2                       &                         & 0.093                   & 0.551                   & 0.004 & 0.116 &  & 0.025 & 0.910 & 0.003 & 0.142 \\
        \bottomrule
      \end{tabular}
      \begin{tablenotes}
        \item Note: This table compares the bias and SD for the OLS estimator and the GMM estimator, with varying sample size and number of pretreatment periods, based on 5000 simulations for each setting.
      \end{tablenotes}
    \end{threeparttable}
  }
\end{center}

\medskip

Table \ref{Ch3_tab_sim3} compares our method with the OLS approach in \cite{hsiao2012panel}. In panel A, both the unobserved individual characteristics and the idiosyncratic shocks are normally distributed so that the linear conditional mean assumption is satisfied. The results show that the GMM estimator outperforms the OLS estimator by having a smaller bias in estimating both the individual treatment effects and the average treatment effect, although the variance of the GMM estimator is also larger.

In panel B, the unobserved individual characteristics are drawn from the uniform distribution, and the linear conditional mean assumption is no longer satisfied \citep{li2017estimation}. We see that the results are virtually unchanged for the GMM estimator, while the OLS estimator performs slightly worse by having larger biases and SD's, which is more pronounced in estimating the individual treatment effects.
The results indicate that the linear conditional mean assumption is not a very strong one. Indeed, the distribution of the sum of several random variables would become more bell-shaped like the normal distribution under fairly general conditions, as a result of the central limit theorem. The simulation results are very similar when the unobserved individual characteristics are drawn from a mix of other distributions.

\begin{center}
  \resizebox{11cm}{!}{
    \begin{threeparttable}[!htbp]
      \centering
      \caption{Simulation}\label{Ch3_tab_sim4}
      \begin{tabular}{ccc@{\hskip 0.7ex}ccccc @{\hskip 0.5ex}ccccc }
        \toprule
        \multicolumn{4}{c}{} & \multicolumn{4}{c}{IFE} & \multicolumn{1}{c}{}    & \multicolumn{4}{c}{GMM}                                                                                                        \\
        \cmidrule(lr){5-8} \cmidrule(lr){10-13}
        \multicolumn{4}{c}{} & \multicolumn{2}{c}{ITT} & \multicolumn{2}{c}{ATT} & \multicolumn{1}{c}{}    & \multicolumn{2}{c}{ITT} & \multicolumn{2}{c}{ATT}                                                    \\
        \cmidrule(lr){5-6} \cmidrule(lr){7-8}\cmidrule(lr){10-11} \cmidrule(lr){12-13}
        $N_1$                & $N_0$                   & $T_0$                   &                         & Bias                    & SD                      & Bias  & SD    &  & Bias  & SD    & Bias  & SD    \\
        \midrule
        \addlinespace[2ex]
        \multicolumn{13}{c}{Panel A: $\boldsymbol{X}_{it}$ constant across $t$}                                                                                                                                   \\
        \addlinespace[1ex]
        5                    & 100                     & 1                       &                         & 1.203                   & 1.357                   & 0.657 & 0.610 &  & 0.047 & 1.499 & 0.016 & 0.687 \\
        5                    & 200                     & 1                       &                         & 1.264                   & 1.229                   & 0.492 & 0.548 &  & 0.089 & 1.624 & 0.023 & 0.730 \\
        5                    & 100                     & 2                       &                         & 0.922                   & 1.140                   & 0.263 & 0.518 &  & 0.034 & 1.265 & 0.015 & 0.585 \\
        5                    & 200                     & 2                       &                         & 0.982                   & 1.147                   & 0.378 & 0.520 &  & 0.040 & 1.387 & 0.018 & 0.634 \\
        \addlinespace[1ex]
        \multicolumn{13}{c}{Panel B: $\boldsymbol{\beta}_{t,k}^0$ constant across $t$}                                                                                                                            \\
        \addlinespace[1ex]
        5                    & 100                     & 1                       &                         & 1.289                   & 1.220                   & 0.773 & 0.579 &  & 0.029 & 1.349 & 0.016 & 0.616 \\
        5                    & 200                     & 1                       &                         & 1.681                   & 1.370                   & 0.836 & 0.620 &  & 0.034 & 1.563 & 0.020 & 0.713 \\
        5                    & 100                     & 2                       &                         & 0.930                   & 1.070                   & 0.440 & 0.486 &  & 0.026 & 1.261 & 0.017 & 0.577 \\
        5                    & 200                     & 2                       &                         & 1.417                   & 1.083                   & 1.015 & 0.489 &  & 0.031 & 1.250 & 0.011 & 0.564 \\
        \bottomrule
      \end{tabular}
      \begin{tablenotes}
        \item Note: This table compares the bias and SD for the IFE estimator and the GMM estimator, with varying sample size and number of pretreatment periods, based on 5000 simulations for each setting.
      \end{tablenotes}
    \end{threeparttable}
  }
\end{center}

\medskip

Table \ref{Ch3_tab_sim4} compares our method with the method of estimating the interactive fixed effects model directly, which was first developed in \cite{bai2009panel} and then adapted into the potential outcomes framework by \cite{xu2017generalized} to allow heterogeneous treatment effects. We fix the number of treated individuals at 5, and compare the performance of the two methods in estimating the individual treatment effect on the treated and the average treatment effect on the treated.

We consider two scenarios that are relevant in the context of empirical microeconomics. In panel A, the observed covariates are constant over time. This is plausible for covariates such as gender, race or education level, which are likely to be stable over time.
Since the IFE method requires the observed covariates to be time-varying, the covariates that are constant over time are dropped from the estimation and become part of the unobserved individual characteristics, which makes the model equivalent to a pure factor model with 4 unobserved factors. As we have 5 related outcomes, this model should still be estimable by the IFE method. However, we see that IFE method perform poorly when there are only a small number of pretreatment outcomes to recover the unobserved individual characteristics. The bias and SD of the IFE estimator become smaller as more pretreatment outcomes are available, but are still quite large compared with our method.

To accommodate the restrictive model specification for the IFE method, we allow the covariates to be time-varying while keeping the coefficients constant over time in panel B, although the coefficients are allowed to vary across outcomes since it is unlikely that the coefficients for different outcomes would be the same in practice.
We see that the IFE estimator has poor performance since the model is still misspecified in their method, whereas the results for our method are virtually unchanged.

\begin{center}
  \resizebox{11cm}{!}{
    \begin{threeparttable}[!htbp]
      \centering
      \caption{Simulation}\label{Ch3_tab_sim5}
      \begin{tabular}{ccc@{\hskip 0.7ex}ccccc @{\hskip 0.5ex}ccccc }
        \toprule
        \multicolumn{4}{c}{} & \multicolumn{4}{c}{SCM} & \multicolumn{1}{c}{}    & \multicolumn{4}{c}{GMM}                                                                                                        \\
        \cmidrule(lr){5-8} \cmidrule(lr){10-13}
        \multicolumn{4}{c}{} & \multicolumn{2}{c}{ITT} & \multicolumn{2}{c}{ATT} & \multicolumn{1}{c}{}    & \multicolumn{2}{c}{ITT} & \multicolumn{2}{c}{ATT}                                                    \\
        \cmidrule(lr){5-6} \cmidrule(lr){7-8}\cmidrule(lr){10-11} \cmidrule(lr){12-13}
        $N_1$                & $N_0$                   & $T_0$                   &                         & Bias                    & SD                      & Bias  & SD    &  & Bias  & SD    & Bias  & SD    \\
        \midrule
        \addlinespace[2ex]
        \multicolumn{13}{c}{Panel A: distributions of $\boldsymbol{\mu}_i$ same for treated and control}                                                                                                          \\
        \addlinespace[1ex]
        5                    & 100                     & 1                       &                         & 0.376                   & 1.247                   & 0.245 & 0.577 &  & 0.029 & 1.349 & 0.016 & 0.617 \\
        5                    & 200                     & 1                       &                         & 0.883                   & 1.376                   & 0.698 & 0.628 &  & 0.035 & 1.563 & 0.020 & 0.713 \\
        5                    & 100                     & 2                       &                         & 0.930                   & 1.191                   & 0.526 & 0.547 &  & 0.023 & 1.243 & 0.014 & 0.565 \\
        5                    & 200                     & 2                       &                         & 0.469                   & 1.186                   & 0.126 & 0.531 &  & 0.031 & 1.250 & 0.012 & 0.564 \\
        \addlinespace[1ex]
        \multicolumn{13}{c}{Panel B: distributions of $\boldsymbol{\mu}_i$ different for treated and control}                                                                                                     \\
        \addlinespace[1ex]
        5                    & 100                     & 1                       &                         & 0.763                   & 1.253                   & 0.634 & 0.605 &  & 0.036 & 1.368 & 0.022 & 0.658 \\
        5                    & 200                     & 1                       &                         & 1.412                   & 1.413                   & 1.269 & 0.656 &  & 0.037 & 1.573 & 0.024 & 0.735 \\
        5                    & 100                     & 2                       &                         & 0.982                   & 1.204                   & 0.513 & 0.558 &  & 0.027 & 1.249 & 0.020 & 0.578 \\
        5                    & 200                     & 2                       &                         & 0.781                   & 1.203                   & 0.613 & 0.551 &  & 0.032 & 1.256 & 0.014 & 0.577 \\
        \bottomrule
      \end{tabular}
      \begin{tablenotes}
        \item Note: This table compares the bias and SD for the SCM estimator and the GMM estimator, with varying sample size and number of pretreatment periods, based on 5000 simulations for each setting.
      \end{tablenotes}
    \end{threeparttable}
  }
\end{center}

\medskip

Table \ref{Ch3_tab_sim5} compares our method with the synthetic control method \citep{abadie2010synthetic}. In panel A, the unobserved individual characteristics for both the treated individuals and the untreated individuals are drawn from $N(0,1)$, while in panel B, the unobserved individual characteristics for the treated individuals are drawn from $N(1,1)$.
Since the synthetic control method requires the treated units to be in the convex hull of the control units by restricting the weights assigned to the control units to be nonnegative, their method may perform poorly when the support of the unobserved individual characteristics are different for the treated and untreated individuals.
While our method should be unaffected by the degree of overlapping in the distributions of the unobserved individual characteristics for the two treatment groups.
The simulation results show that indeed the synthetic control estimator performs worse in panel B. Perhaps somewhat surprising is that its performance is also poor compared with our method in panel A. This is because the coefficients are outcome-specific, so that the levels of the outcomes are also likely to vary across outcomes, which makes it more difficult to obtain a good pretreatment fit under the nonnegativity restriction. In comparison, our method has good performance in both panels.

Overall, the simulation results show that our method has good performance in terms of the bias and SD in estimating the individual treatment effects and the average treatment effect under various settings, and has superior performance than related methods. The shortcoming of our method is that the confidence intervals tend to be too wide, especially when the number of pretreatment outcomes is small.

\section{Empirical Application}\label{Ch3_sec_app}

We illustrate our method by estimating the effect of health insurance coverage on the individual usage of hospital emergency departments.

Although the usage of emergency departments applies to only a small proportion of the population, it imposes great financial pressure on the health care system. In addition, it is not clear ex ante what the direction of the effect should be.
E.g., \cite{taubman2014medicaid} argues that health insurance coverage could either increase emergency-department use by reducing its cost for the patients, or decrease emergency-department use by encouraging primary care use or improving health.

The findings on emergency-department use have been mixed.
Using survey data collected from the participants of the Oregon Health Insurance Experiment (OHIE) about a year after they were notified of the selection results, \cite{finkelstein2012oregon} find no discernible impact of health insurance coverage on emergency-department use.\footnote{The Oregon Health Insurance Experiment (OHIE) was initiated in 2008, targeting at low-income adults in Oregon who had been without health insurance for at least 6 months. Among the 89,824 individuals who signed up, 35,169 individuals were randomly selected by the lottery and were eligible to apply for the Oregon Health Plan (OHP) Standard program, which provided relatively comprehensive medical benefits with no consumer cost sharing, and the monthly premiums was only between \$0 and \$20 depending on the income. As a randomised controlled experiment, the OHIE offers an opportunity for researchers to study the effect of health insurance coverage on various health outcomes without confounding factors.}
While using the visit-level data for all emergency-department visits to twelve hospitals in the Portland area probabilistically matched to the OHIE study population on the basis of name, date of birth, and gender, \cite{taubman2014medicaid} find that health insurance coverage significantly increases emergency-department use by 0.41 visits per person, from an average of 1.02 visits per person in the control group in the first 15 months of the experiment. They also examine whether the effect differs across heterogeneous groups, and find statistically significant increases in emergency-department use across most subgroups in terms of the number of pre-experiment emergency-department visits, hospital admission (inpatient or outpatient visits), timing (on-hours or off-hours visits), the type of visits (emergent and not preventable, emergent and preventable, primary care treatable, and non-emergent), as well as gender, age, and health condition.

In this application, we wish to estimate the effect of health insurance coverage on emergency-department use for each individual in the sample.
% , and identify those for whom the effect is statistically significant
This would potentially help us better understand whether and how health insurance coverage affects emergency-department use, compared with using only the average treatment effect for the whole sample or for some preassigned subgroups (conditional average treatment effects).

Our data combines both the hospital emergency-department visit-level data and the survey data. There are two time periods, one before the randomisation and one after.\footnote{The pre-randomisation period in the hospital visit-level data was from January 2007 to March 2008, and the post-randomisation period was from March 2008 to September 2009. The two surveys were collected shortly after the randomisation and about a year after randomisation, respectively, each covering a 6-month period before the survey.}
To estimate the individual treatment effects, we include 3 observed covariates including gender, birth year, and household income as a percentage of the federal poverty line, and 10 related outcomes including different types of emergency-department visits and medical charges.
We also consider a rich list of variables on which we make comparisons for individuals with different estimated treatment effects.
There are 2154 individuals with complete information on these variables.\footnote{Note that our sample size is significantly smaller than the other studies using the OHIE data, due to the inclusion of the extensive list of variables. For example, the sample size in \cite{finkelstein2012oregon} is 74,922, and the sample size in \cite{taubman2014medicaid} is 24,646. So our sample may not be representative of the OHIE sample and the results in different studies may not be directly comparable.}

\resizebox{\columnwidth}{!}{
  \begin{threeparttable}
    \centering
    \caption{Sample Selection}\label{Ch3_tab_sel}
    \begin{tabular}{rrrcrrc}
      \hline
                                                          & Selected & Not-selected & Difference & Insured & Not-insured & Difference \\
                                                          & (1)      & (2)          & (3)        & (4)     & (5)         & (6)        \\
      \hline
      Female                                              & 0.59     & 0.60         & -0.01      & 0.63    & 0.58        & 0.05*      \\
      Birth year                                          & 1966.24  & 1966.44      & -0.19      & 1967.03 & 1966.09     & 0.95       \\
      Household income as percent of federal poverty line & 79.77    & 75.67        & 4.10       & 53.73   & 86.57       & -32.84***  \\
      \# ED visits                                        & 0.32     & 0.43         & -0.11**    & 0.48    & 0.33        & 0.15**     \\
      \# outpatient ED visits                             & 0.27     & 0.35         & -0.09**    & 0.40    & 0.28        & 0.13**     \\
      \# weekday daytime ED visits                        & 0.18     & 0.24         & -0.06**    & 0.27    & 0.19        & 0.08**     \\
      \# emergent non-preventable ED visits               & 0.07     & 0.09         & -0.02      & 0.11    & 0.07        & 0.04**     \\
      \# emergent preventable ED visits                   & 0.03     & 0.03         & 0.00       & 0.03    & 0.02        & 0.01       \\
      \# primary care treatable ED visits                 & 0.10     & 0.15         & -0.05***   & 0.15    & 0.12        & 0.04       \\
      Total charges                                       & 859.71   & 1276.90      & -417.19*   & 1379.58 & 947.54      & 432.04     \\
      Total ED charges                                    & 345.48   & 504.64       & -159.16**  & 494.95  & 396.86      & 98.09      \\
      \# ED visits to a high uninsured volume hospital    & 0.17     & 0.22         & -0.05      & 0.25    & 0.17        & 0.08**     \\
      \# ED visits (survey)                               & 0.24     & 0.30         & -0.06*     & 0.38    & 0.23        & 0.15***    \\
      \hline
      N                                                   & 1103     & 1051         &            & 577     & 1577        &            \\
      \hline
    \end{tabular}
    \begin{tablenotes}
      \item 1) This table compares the mean values of the covariates and related outcomes in the pretreatment period for individuals selected/not-selected by the lottery, and individuals insured/not-insured.
      \item 2) Significance levels of the two-sample t-test: * 10\%, ** 5\%, *** 1\%.
    \end{tablenotes}
  \end{threeparttable}
}

\medskip

The first 3 columns in Table \ref{Ch3_tab_sel} present the mean values of the covariates and outcomes in the pretreatment period for individuals selected by the lottery and for individuals not selected by the lottery, as well as the difference between the two groups.
Since a considerable number of observations with incomplete information are dropped, being selected by the lottery is negatively correlated with different types of emergency-department visits in the pretreatment period in our sample, which suggests that the lottery assignment is not likely to be a valid instrument for health insurance coverage.
Table \ref{Ch3_tab_sel} also compares the mean pretreatment characteristics for individuals covered by health insurance and those not covered, which shows that individuals who were covered were poorer and used emergency-department more frequently in the pretreatment period than people who were not covered by health insurance.

\begin{figure}[!htbp]
  \centering
  \includegraphics[width=.7\textwidth]{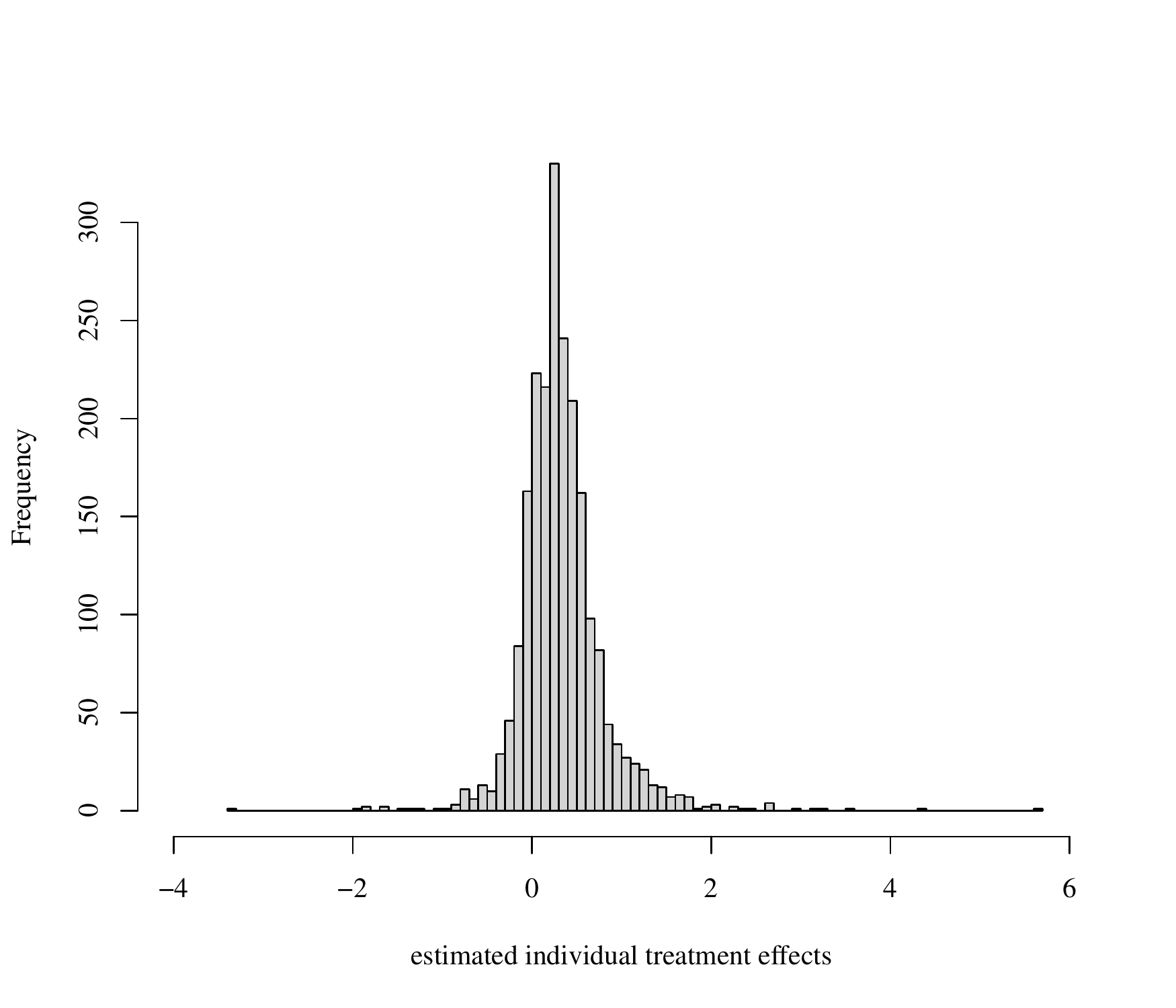}
  \caption{Distribution of the estimated individual treatment effects}
  \label{Ch3_fig_ITE}
\end{figure}

Figure \ref{Ch3_fig_ITE} shows the distribution of the estimated individual treatment effects using our method.\footnote{As mentioned earlier, since only individuals with complete information on the variables are selected into our sample, the distribution may not be representative of the OHIE participants. The distribution of the estimated individual treatment effects may be more spread out than the distribution of the true effects due to noise, or less spread out since the estimates are based on the parametric models for the potential outcomes, which may be over-simplifying compared with the true models.} The mean of estimated individual treatment effects, or the estimated average treatment effect is 0.33, which is significant at 1\% level. 114 individuals have treatment effects that are significant at 10\% level, among which 23 are negative and 91 are positive.\footnote{Note that if we were to adjust for multiple testing, e.g., using the Benjamini–Hochberg procedure to control the false discovery rate (FDR) at 10\% level, then we would be left with only one individual whose treatment effect is significant. Although the small number of individuals with significant treatment effects may also be attributed to the overestimation of the variance of the individual treatment effect estimator.}
We then move on to compare the characteristics of the individuals based on their estimated treatment effects, which are presented in Table \ref{Ch3_tab_comp1}-\ref{Ch3_tab_comp3}.
Column (1) shows the mean characteristics of individuals whose treatment effects are not significant at 10\% level, column (2) shows the mean characteristics of individuals whose treatment effects are significantly negative, column (3) shows the differences between column (2) and column (1), column (4) shows the mean characteristics of individuals whose treatment effects are significantly positive, and column (5) shows the differences between column (4) and column (1).

Compared with individuals who would not be significantly affected by the treatment, individuals who would significantly decrease or increase their emergency-department visits if covered by health insurance both had more emergency-department visits and more medical charges in the pretreatment period. However, these two groups were also distinct in some characteristics.

The individuals who would have fewer emergency-department visits if covered by health insurance were on average 7 years younger than individuals in the control group and 10 years younger than the positive group, more likely to be female with less education, and in particular, were much poorer than individuals in the other groups. They were more likely to be diagnosed with depression but not other conditions. Importantly, they were less likely to have any primary care visits, and more likely to use emergency department as the place for medical care.
In terms of emergency-department use in the pretreatment period, they had fewer visits resulting in hospitalisation, more outpatient visits, more preventable and non-emergent visits, more visits to hospitals with a low fraction of uninsured patients, fewer visits for chronic conditions, and more visits for injury. Although their medical charges were not as high as those for individuals in the positive group, they owed more money for medical expenses.

In comparison, individuals who would have more emergency-department visits if covered by health insurance were more likely to be older, male, and with household income right above the federal poverty line, which means that they were not as poor as the individuals in the other groups.
They were in worse health conditions, more likely to be diagnosed with diabetes and high blood pressure, and were taking more prescription medications.
They also had more emergency-department visits of all types in the pretreatment period, including visits resulting in hospitalisation and visits for more severe conditions such as chronic conditions, chest pain and psychological conditions, and they incurred more medical charges.

Overall, these comparisons suggest that the individuals who would have fewer emergency-department visits if covered by health insurance were younger and not in very bad physical conditions. However, their access to primary care were limited due to being in much more disadvantaged positions financially, which made them resort to using the emergency department as the usual place for medical care. In contrast, the individuals who would have more emergency-department visits if covered by health insurance were more likely to be older and in poor health. So even with access to primary care, they still used emergency departments more often for severe conditions, although sometimes for primary care treatable and non-emergent conditions as well.

All in all, it seems that both mechanisms discussed by \cite{taubman2014medicaid} are playing a role. For people who used emergency department for medical care because they did not have access to primary care service, health insurance coverage decreases emergency-department use because it increases access to primary care and may also lead to improved health. Whereas for people who had access to primary care and still used emergency department due to worse physical conditions, health insurance coverage increases emergency-department use because it reduces the out-of-pocket cost of the visits.

This application shows the potential value of estimating individual treatment effects for policy evaluation. Our findings would not have been possible by only estimating conditional average treatment effects, as we would not be able to distinguish individuals with positive or negative treatment effects at the first place.

% Since individuals with negative or positive treatment effects may share many similar characteristics, when put together, we may not be able to see a significant overall effect for the whole sample or subgroups divided according to these characteristics.
% By estimating the individual treatment effects, we may better distinguish individuals with different treatment effects, which could then guide us in designing individualised policies that target these groups.

\resizebox{\columnwidth}{!}{
  \begin{threeparttable}
    \centering
    \caption{Comparison of Characteristics}\label{Ch3_tab_comp1}
    \begin{tabular}{rrrcrc}
      \hline
                                                          & Same    & Fewer   & Difference & More    & Difference \\
                                                          & (1)     & (2)     & (3)        & (4)     & (5)        \\

      \hline
      Birth year                                          & 1966.42 & 1973.90 & 7.48*      & 1964.09 & -2.33**    \\
      Female                                              & 0.61    & 1.00    & 0.39***    & 0.37    & -0.24***   \\
      Education                                           & 2.52    & 1.90    & -0.62**    & 2.30    & -0.22**    \\
      English                                             & 0.89    & 0.90    & 0.01       & 0.95    & 0.07***    \\
      {Race:}                                                                                                     \\
      \textit{White}                                      & 0.75    & 0.50    & -0.25      & 0.79    & 0.04       \\
      \textit{Hispanic}                                   & 0.10    & 0.10    & 0.00       & 0.12    & 0.02       \\
      \textit{Black}                                      & 0.06    & 0.30    & 0.24       & 0.06    & 0.00       \\
      \textit{Asian}                                      & 0.10    & 0.10    & 0.00       & 0.06    & -0.05*     \\
      \textit{American Indian or Alaska Native}           & 0.04    & 0.10    & 0.06       & 0.07    & 0.03       \\
      \textit{Native Hawaiian or Pacific Islander}        & 0.01    & 0.00    & -0.01***   & 0.01    & 0.00       \\
      \textit{Other races}                                & 0.07    & 0.10    & 0.03       & 0.06    & -0.01      \\
      Employed                                            & 0.53    & 0.60    & 0.07       & 0.48    & -0.05      \\
      Average hours worked per week                       & 2.25    & 2.40    & 0.15       & 2.20    & -0.05      \\
      Household income as percent of federal poverty line & 76.12   & 28.02   & -48.09***  & 114.98  & 38.86***   \\
      Household Size (adults and children)                & 2.99    & 3.50    & 0.51       & 2.61    & -0.39**    \\
      Number of family members under 19 living in house   & 0.88    & 1.00    & 0.12       & 0.56    & -0.33***   \\
      Overall health                                      & 3.01    & 2.40    & -0.61      & 2.84    & -0.17      \\
      Health change                                       & -0.09   & -0.10   & -0.01      & -0.08   & 0.01       \\
      \# days physical health not good                    & 6.97    & 7.60    & 0.63       & 9.49    & 2.52**     \\
      \# days mental health not good                      & 8.43    & 13.90   & 5.47       & 10.30   & 1.87       \\
      \# days poor health impaired regular activities     & 6.09    & 7.90    & 1.81       & 8.36    & 2.27**     \\
      Diabetes                                            & 0.10    & 0.20    & 0.10       & 0.20    & 0.10**     \\
      Asthma                                              & 0.13    & 0.20    & 0.07       & 0.13    & 0.01       \\
      High blood pressure                                 & 0.24    & 0.20    & -0.04      & 0.38    & 0.15***    \\
      Depression                                          & 0.37    & 0.70    & 0.33*      & 0.47    & 0.10**     \\
      Any primary care visits                             & 0.55    & 0.20    & -0.35**    & 0.57    & 0.02       \\
      \# primary care visits                              & 1.64    & 1.80    & 0.16       & 1.86    & 0.22       \\
      Any hospital visits                                 & 0.04    & 0.10    & 0.06       & 0.12    & 0.07**     \\
      \# hospital visits                                  & 0.05    & 0.10    & 0.05       & 0.16    & 0.11**     \\
      {Usual place for medical care:}                                                                             \\
      \textit{Private clinic}                             & 0.18    & 0.00    & -0.18***   & 0.28    & 0.10**     \\
      \textit{Public clinic}                              & 0.17    & 0.10    & -0.07      & 0.18    & 0.01       \\
      \textit{Hospital-based clinic}                      & 0.07    & 0.00    & -0.07***   & 0.08    & 0.01       \\
      \textit{Hospital ER}                                & 0.03    & 0.50    & 0.47**     & 0.06    & 0.03       \\
      \textit{Urgent care clinic}                         & 0.03    & 0.10    & 0.07       & 0.02    & -0.01      \\
      \textit{Other places}                               & 0.06    & 0.00    & -0.06***   & 0.07    & 0.00       \\
      \textit{Don't have usual place}                     & 0.46    & 0.30    & -0.16      & 0.32    & -0.14***   \\
      \hline
    \end{tabular}
    \begin{tablenotes}
      \item 1) This table shows the mean characteristics for individuals whose treatment effects are not statistically significant at 10\% level (column 1), whose treatment effects are significantly negative (column 2), and whose treatment effects are significantly positive (column 4).
      Column (3) contains the differences between column (2) and column (1), and column (5) contains the differences between column (4) and column (1).
      \item 2) Significance levels of the two-sample t-test: * 10\%, ** 5\%, *** 1\%.
    \end{tablenotes}
  \end{threeparttable}
}

\resizebox{\columnwidth}{!}{
  \begin{threeparttable}
    \centering
    \caption{Comparison of Characteristics}\label{Ch3_tab_comp2}
    \begin{tabular}{rrrcrc}
      \hline
                                                            & Same    & Fewer   & Difference & More    & Difference \\
                                                            & (1)     & (2)     & (3)        & (4)     & (5)        \\

      \hline
      Needed medical care                                   & 0.68    & 1.00    & 0.32***    & 0.76    & 0.08*      \\
      Got all needed medical care                           & 0.58    & 0.40    & -0.18      & 0.56    & -0.02      \\
      Reason went without care:                                                                                     \\
      \textit{Cost too much}                                & 0.33    & 0.20    & -0.13      & 0.31    & -0.02      \\
      \textit{No insurance}                                 & 0.36    & 0.50    & 0.14       & 0.35    & -0.01      \\
      \textit{Doc wouldn't take insurance}                  & 0.01    & 0.00    & -0.01***   & 0.00    & -0.01***   \\
      \textit{Owed money to provider}                       & 0.05    & 0.00    & -0.05***   & 0.10    & 0.05*      \\
      \textit{Couldn't get an appointment}                  & 0.03    & 0.00    & -0.03***   & 0.01    & -0.02*     \\
      \textit{Office wasn't open}                           & 0.01    & 0.00    & -0.01***   & 0.00    & -0.01***   \\
      \textit{Didn't have a doctor}                         & 0.11    & 0.20    & 0.09       & 0.10    & -0.01      \\
      \textit{Other reasons}                                & 0.03    & 0.00    & -0.03***   & 0.04    & 0.01       \\
      \textit{Don't know}                                   & 0.00    & 0.00    & 0.00***    & 0.00    & 0.00***    \\
      Needed prescription medications                       & 0.61    & 0.80    & 0.19       & 0.77    & 0.16***    \\
      Got all needed prescriptions                          & 0.74    & 0.70    & -0.04      & 0.64    & -0.09*     \\
      Currently taking any prescription medications         & 0.44    & 0.30    & -0.14      & 0.68    & 0.25***    \\
      \# prescription medications taking                    & 1.37    & 1.30    & -0.07      & 2.56    & 1.18***    \\
      Reason went without prescription medication:                                                                  \\
      \textit{Cost too much}                                & 0.21    & 0.20    & -0.01      & 0.27    & 0.06       \\
      \textit{No insurance}                                 & 0.20    & 0.10    & -0.10      & 0.20    & 0.00       \\
      \textit{Didn't have doctor}                           & 0.08    & 0.10    & 0.02       & 0.10    & 0.01       \\
      \textit{Couldn't get prescription}                    & 0.08    & 0.10    & 0.02       & 0.07    & -0.01      \\
      \textit{Couldn't get to pharmacy}                     & 0.01    & 0.10    & 0.09       & 0.00    & -0.01***   \\
      \textit{Other reasons}                                & 0.02    & 0.00    & -0.02***   & 0.04    & 0.02       \\
      \textit{Don't know}                                   & 0.00    & 0.00    & 0.00       & 0.00    & 0.00       \\
      Needed dental care                                    & 0.70    & 1.00    & 0.30***    & 0.70    & 0.00       \\
      Got all needed dental care                            & 0.41    & 0.30    & -0.11      & 0.37    & -0.05      \\
      Any ER visits                                         & 0.14    & 0.90    & 0.76***    & 0.26    & 0.12***    \\
      \# of ER visits                                       & 0.25    & 2.40    & 2.15***    & 0.45    & 0.20**     \\
      Used emergency room for non-emergency care            & 0.02    & 0.10    & 0.08       & 0.06    & 0.03       \\
      Reason went to ER:                                                                                            \\
      \textit{Needed emergency care}                        & 0.05    & 0.80    & 0.75***    & 0.06    & 0.01       \\
      \textit{Clinics closed}                               & 0.01    & 0.30    & 0.29*      & 0.02    & 0.00       \\
      \textit{Couldn't get doctor's appointment}            & 0.02    & 0.20    & 0.18       & 0.02    & 0.00       \\
      \textit{Didn't have personal doctor}                  & 0.02    & 0.30    & 0.28       & 0.03    & 0.01       \\
      \textit{Couldn't afford copay to see a doctor}        & 0.01    & 0.20    & 0.19       & 0.02    & 0.00       \\
      \textit{Didn't know where else to go}                 & 0.02    & 0.20    & 0.18       & 0.04    & 0.02       \\
      \textit{Other reason}                                 & 0.01    & 0.10    & 0.09       & 0.02    & 0.01       \\
      \textit{Needed prescription drug}                     & 0.01    & 0.10    & 0.09       & 0.00    & -0.01***   \\
      \textit{Don't know}                                   & 0.00    & 0.00    & 0.00       & 0.00    & 0.00       \\
      Any out of pocket costs for medical care              & 0.65    & 0.80    & 0.15       & 0.71    & 0.06       \\
      Total out of pocket costs for medical care            & 5195.07 & 1257.00 & -3938.07   & 1136.44 & -4058.62   \\
      Borrowed money/skipped bills to pay health care bills & 0.34    & 0.50    & 0.16       & 0.47    & 0.13**     \\
      Currently owe money for medical expenses              & 0.46    & 0.80    & 0.34**     & 0.71    & 0.25***    \\
      Total amount currently owed for medical expenses      & 1559.40 & 7354.00 & 5794.60**  & 5694.17 & 4134.77*** \\
      \hline
    \end{tabular}
    \begin{tablenotes}
      \item 1) This table shows the mean characteristics for individuals whose treatment effects are not statistically significant at 10\% level (column 1), whose treatment effects are significantly negative (column 2), and whose treatment effects are significantly positive (column 4).
      Column (3) contains the differences between column (2) and column (1), and column (5) contains the differences between column (4) and column (1).
      \item 2) Significance levels of the two-sample t-test: * 10\%, ** 5\%, *** 1\%.
    \end{tablenotes}
  \end{threeparttable}
}

\resizebox{\columnwidth}{!}{
  \begin{threeparttable}
    \centering
    \caption{Comparison of Characteristics}\label{Ch3_tab_comp3}
    \begin{tabular}{rrrcrc}
      \hline
                                                         & Same   & Fewer   & Difference & More    & Difference \\
                                                         & (1)    & (2)     & (3)        & (4)     & (5)        \\

      \hline
      Any ED visits                                      & 0.16   & 0.90    & 0.74***    & 0.84    & 0.68***    \\
      \# ED visits                                       & 0.28   & 3.10    & 2.82***    & 1.94    & 1.67***    \\
      Any ED visits resulting in hospitalization         & 0.03   & 0.00    & -0.03***   & 0.37    & 0.34***    \\
      \# ED visits resulting in hospitalization          & 0.03   & 0.00    & -0.03***   & 0.62    & 0.59***    \\
      Any outpatient ED visits                           & 0.15   & 0.90    & 0.75***    & 0.67    & 0.52***    \\
      \# outpatient ED visits                            & 0.25   & 3.10    & 2.85***    & 1.32    & 1.07***    \\
      Any weekday daytime ED visits                      & 0.10   & 0.60    & 0.50**     & 0.62    & 0.52***    \\
      \# weekday daytime ED visits                       & 0.15   & 1.50    & 1.35**     & 1.14    & 0.99***    \\
      Any off-time ED visits                             & 0.09   & 0.90    & 0.81***    & 0.51    & 0.42***    \\
      \# off-time ED visits                              & 0.12   & 1.60    & 1.48***    & 0.83    & 0.71***    \\
      \# emergent non-preventable ED visits              & 0.05   & 0.17    & 0.12       & 0.66    & 0.61***    \\
      \# emergent preventable ED visits                  & 0.02   & 0.14    & 0.12**     & 0.15    & 0.12***    \\
      \# primary care treatable ED visits                & 0.10   & 1.37    & 1.27***    & 0.46    & 0.36***    \\
      \# non-emergent ED visits                          & 0.05   & 1.22    & 1.16***    & 0.35    & 0.29***    \\
      \# unclassified ED visits                          & 0.05   & 0.20    & 0.15       & 0.36    & 0.30***    \\
      Any ambulatory case sensitive ED visits            & 0.01   & 0.00    & -0.01***   & 0.15    & 0.14***    \\
      \# ambulatory case sensitive ED visits             & 0.02   & 0.00    & -0.02***   & 0.15    & 0.13***    \\
      Any ED visits to a high uninsured volume hospital  & 0.08   & 0.20    & 0.12       & 0.76    & 0.68***    \\
      \# ED visits to a high uninsured volume hospital   & 0.12   & 0.30    & 0.18       & 1.61    & 1.49***    \\
      Any ED visits to a low uninsured volume hospital   & 0.09   & 0.90    & 0.81***    & 0.19    & 0.10**     \\
      \# ED visits to a low uninsured volume hospital    & 0.15   & 2.80    & 2.65***    & 0.33    & 0.17**     \\
      Any ED visits for chronic conditions               & 0.03   & 0.30    & 0.27       & 0.36    & 0.32***    \\
      \# ED visits for chronic conditions                & 0.05   & 0.50    & 0.45       & 0.62    & 0.57***    \\
      Any ED visits for injury                           & 0.06   & 0.40    & 0.34*      & 0.32    & 0.26***    \\
      \# ED visits for injury                            & 0.07   & 0.40    & 0.33*      & 0.43    & 0.37***    \\
      Any ED visits for skin conditions                  & 0.01   & 0.10    & 0.09       & 0.02    & 0.01       \\
      \# ED visits for skin conditions                   & 0.02   & 0.10    & 0.08       & 0.03    & 0.01       \\
      Any ED visits for abdominal pain                   & 0.01   & 0.10    & 0.09       & 0.06    & 0.05**     \\
      \# ED visits for abdominal pain                    & 0.01   & 0.10    & 0.09       & 0.10    & 0.09**     \\
      Any ED visits for back pain                        & 0.01   & 0.30    & 0.29*      & 0.04    & 0.03       \\
      \# ED visits for back pain                         & 0.01   & 0.40    & 0.39       & 0.06    & 0.04       \\
      Any ED visits for chest pain                       & 0.01   & 0.00    & -0.01***   & 0.05    & 0.04*      \\
      \# ED visits for chest pain                        & 0.01   & 0.00    & -0.01***   & 0.05    & 0.04*      \\
      Any ED visits for headache                         & 0.01   & 0.00    & -0.01***   & 0.01    & 0.00       \\
      \# ED visits for headache                          & 0.01   & 0.00    & -0.01***   & 0.01    & 0.00       \\
      Any ED visits for mood disorders                   & 0.00   & 0.00    & 0.00**     & 0.09    & 0.08***    \\
      \# ED visits for mood disorders                    & 0.00   & 0.00    & 0.00**     & 0.15    & 0.15**     \\
      Any ED visits for psych conditions/substance abuse & 0.01   & 0.00    & -0.01***   & 0.17    & 0.17***    \\
      \# ED visits for psych conditions/substance abuse  & 0.01   & 0.00    & -0.01***   & 0.36    & 0.34***    \\
      Total ED charges                                   & 274.98 & 1818.44 & 1543.46**  & 3195.22 & 2920.24*** \\
      Total charges                                      & 639.70 & 2223.85 & 1584.16**  & 9260.22 & 8620.52*** \\
      \hline
    \end{tabular}
    \begin{tablenotes}
      \item 1) This table shows the mean characteristics for individuals whose treatment effects are not statistically significant at 10\% level (column 1), whose treatment effects are significantly negative (column 2), and whose treatment effects are significantly positive (column 4).
      Column (3) contains the differences between column (2) and column (1), and column (5) contains the differences between column (4) and column (1).
      \item 2) Significance levels of the two-sample t-test: * 10\%, ** 5\%, *** 1\%.
    \end{tablenotes}
  \end{threeparttable}
}

\section{Conclusion}\label{Ch3_sec_con}

In this paper, we propose a method for estimating the individual treatment effects using panel data, where multiple related outcomes are observed for a large number of individuals over a small number of pretreatment periods. The method is based on the interactive fixed effects model, and allows both the treatment assignment and the potential outcomes to be correlated with the unobserved individual characteristics. Monte Carlo simulations show that our method outperforms related methods. We also provide an example of estimating the effect of health insurance coverage on individual usage of hospital emergency departments using the Oregon Health Insurance Experiment data.

There are several directions for future research.
First, our method requires the idiosyncratic shocks in the pretreatment outcomes to be uncorrelated either over time or across outcomes. It would be a valuable addition to allow (or detect and adjust for) more general dependence structure in the idiosyncratic shocks.
% use overidentification test (J test) to test if there is unwanted correlation?
Second, since the residuals of the rearranged models are not estimates of the idiosyncratic shocks, the variance of our estimator may be over-estimated, especially when the number of pretreatment outcomes is small. A necessary step for future research is to correct for this bias.
Third, the repeated pretreatment set splitting and averaging approach in our method is computationally expensive. It would be an interesting direction for future research to find better ways to select related outcomes or use more flexible averaging scheme.
Fourth, the linear model specification may be restrictive. There is potential to extend our method, perhaps in combination with more flexible machine learning methods, to work with more general nonlinear outcomes.

\begin{appendices}

  \section{Proofs}

  \begin{proof}[Proof of Proposition \ref{Ch3_prop_ITE_GMM}]

    \begin{align*}
      \widehat{\tau}_{it,k}-\tau_{it,k}
      = & \widehat{Y}_{it,k}^1-\widehat{Y}_{it,k}^0-\left(Y_{it,k}^1-Y_{it,k}^0\right)                                                                                                                                                 \\
      = & \boldsymbol{Z}_{it}'\left(\widehat{\boldsymbol{\theta}}_{t,k}^1-\boldsymbol{\theta}_{t,k}^1\right)-\boldsymbol{Z}_{it}'\left(\widehat{\boldsymbol{\theta}}_{t,k}^0-\boldsymbol{\theta}_{t,k}^0\right)-e_{it,k}^1+e_{it,k}^0.
    \end{align*}

    Given the assumptions and our models in \eqref{Ch3_eq_IFE1} and \eqref{Ch3_eq_IFE0}, we have that $Y_{it,k}$, $t\le T_0$, $k\in\mathcal{K}$ are i.i.d. for all $i\in\mathcal{T}$ and all $i\in\mathcal{C}$, and $\mathbb{E}|Y_{it,k}|^2<\infty$, so that $\boldsymbol{Z}_{it}$ and $\boldsymbol{R}_{it}$ are i.i.d. for all $i\in\mathcal{T}$ and all $i\in\mathcal{C}$, $\mathbb{E}\Vert\boldsymbol{Z}_{it}\Vert^2<\infty$, and $\mathbb{E}\Vert\boldsymbol{R}_{it}\Vert^2<\infty$. In addition, $\mathbb{E}\left(\boldsymbol{Z}_{it}\boldsymbol{R}_{it}'\right)$ has full rank due to the observed covariates and the unobserved individual characteristics, so that $\mathbb{E}\left(\boldsymbol{Z}_{it}\boldsymbol{R}_{it}'\right)\boldsymbol{W}^1\mathbb{E}\left(\boldsymbol{R}_{it}\boldsymbol{Z}_{it}'\right)$ is invertible.
    Therefore, the weak law of large numbers and the continuous mapping theorem hold, and
    \begin{align*}
      \widehat{\boldsymbol{\theta}}_{t,k}^1-\boldsymbol{\theta}_{t,k}^1 & =\left({\boldsymbol{Z}_{t}^1}'\boldsymbol{R}_{t}^1\boldsymbol{W}^1{\boldsymbol{R}_{t}^1}'\boldsymbol{Z}_{t}^1\right)^{-1}{\boldsymbol{Z}_{t}^1}'\boldsymbol{R}_{t}^1\boldsymbol{W}^1{\boldsymbol{R}_{t}^1}'\boldsymbol{e}_{t,k}^1                                                                                        \\
                                                                        & =\left(\frac{1}{N_1}{\boldsymbol{Z}_{t}^1}'\boldsymbol{R}_{t}^1\boldsymbol{W}^1\frac{1}{N_1}{\boldsymbol{R}_{t}^1}'\boldsymbol{Z}_{t}^1\right)^{-1}\frac{1}{N_1}{\boldsymbol{Z}_{t}^1}'\boldsymbol{R}_{t}^1\boldsymbol{W}^1\frac{1}{N_1}{\boldsymbol{R}_{t}^1}'\boldsymbol{e}_{t,k}^1                                    \\
                                                                        & \overset{p}{\rightarrow}\left[\mathbb{E}\left(\boldsymbol{Z}_{it}\boldsymbol{R}_{it}'\right)\boldsymbol{W}^1\mathbb{E}\left(\boldsymbol{R}_{it}\boldsymbol{Z}_{it}'\right)\right]^{-1}\mathbb{E}\left(\boldsymbol{Z}_{it}\boldsymbol{R}_{it}'\right)\boldsymbol{W}^1\mathbb{E}\left(\boldsymbol{R}_{it}e_{it,k}^1\right) \\
                                                                        & =0,
    \end{align*}
    as $N_1\rightarrow\infty$.
    Similarly, it can be shown that $\widehat{\boldsymbol{\theta}}_{t,k}^0-\boldsymbol{\theta}_{t,k}^0\overset{p}{\rightarrow}0$ as $N_0\rightarrow\infty$.

    Since $\boldsymbol{Z}_{it}=O_p(1)$, we have $\boldsymbol{Z}_{it}'\left(\widehat{\boldsymbol{\theta}}_{t,k}^1-\boldsymbol{\theta}_{t,k}^1\right)=o_p(1)$ and $\boldsymbol{Z}_{it}'\left(\widehat{\boldsymbol{\theta}}_{t,k}^0-\boldsymbol{\theta}_{t,k}^0\right)=o_p(1)$.
    We also have $\mathbb{E}\left(e_{it,k}^1\mid \boldsymbol{H}_{it}=\boldsymbol{h}_{it}\right)=0$ and $\mathbb{E}\left(e_{it,k}^0\mid \boldsymbol{H}_{it}=\boldsymbol{h}_{it}\right)=0$ under Assumption \ref{Ch3_assume_error}.

    \medskip

    Under the assumptions and by the Cauchy-Schwarz inequality, there exists $M^*\in[0,\infty)$ such that $\mathbb{E}\left\vert\boldsymbol{Z}_{it}'\left(\widehat{\boldsymbol{\theta}}_{t,k}^1-\boldsymbol{\theta}_{t,k}^1\right)\right\vert\le\left(\mathbb{E}\left\Vert\boldsymbol{Z}_{it}\right\Vert^2\right)^{1/2}\left(\mathbb{E}\left\Vert\widehat{\boldsymbol{\theta}}_{t,k}^1-\boldsymbol{\theta}_{t,k}^1\right\Vert^2\right)^{1/2}<M^*$,
    $\mathbb{E}\left\vert\boldsymbol{Z}_{it}'\left(\widehat{\boldsymbol{\theta}}_{t,k}^0-\boldsymbol{\theta}_{t,k}^0\right)\right\vert<M^*$,
    and $\mathbb{E}\left\vert e_{it,k}^0-e_{it,k}^1\right\vert<M^*$.
    By the triangle inequality, $\mathbb{E}\left\vert\widehat{\tau}_{it,k}-\tau_{it,k}\right\vert\le\mathbb{E}\left\vert\boldsymbol{Z}_{it}'\left(\widehat{\boldsymbol{\theta}}_{t,k}^1-\boldsymbol{\theta}_{t,k}^1\right)\right\vert+\mathbb{E}\left\vert\boldsymbol{Z}_{it}'\left(\widehat{\boldsymbol{\theta}}_{t,k}^0-\boldsymbol{\theta}_{t,k}^0\right)\right\vert+\mathbb{E}\left\vert e_{it,k}^0-e_{it,k}^1\right\vert<3M^*$, which implies that $\widehat{\tau}_{it,k}-\tau_{it,k}$ is uniformly integrable. Then by Lebesgue's Dominated Convergence Theorem, convergence in probability implies convergence in means, i.e.,
    \begin{align*}
        & \underset{N_1,N_0\rightarrow\infty}{\text{lim}}\mathbb{E}\left(\widehat{\tau}_{it,k}-\tau_{it,k}\mid \boldsymbol{H}_{it}=\boldsymbol{h}_{it}\right)               \\
      = & \mathbb{E}\left[\underset{N_1,N_0\rightarrow\infty}{\text{plim}}\left(\widehat{\tau}_{it,k}-\tau_{it,k}\right)\mid \boldsymbol{H}_{it}=\boldsymbol{h}_{it}\right] \\
      = & \mathbb{E}\left(e_{it,k}^0-e_{it,k}^1\mid \boldsymbol{H}_{it}=\boldsymbol{h}_{it}\right)                                                                          \\
      = & 0.
    \end{align*}

  \end{proof}

  \begin{proof}[Proof of Proposition \ref{Ch3_prop_ATE_GMM}]
    Under Assumptions \ref{Ch3_assume_predictor}-\ref{Ch3_assume_rank}, central limit theorem applies, and we have
    \begin{align*}
        & \frac{1}{N}\sum_{i=1}^N\left(\widehat{\tau}_{it,k}-\tau_{it,k}\right)                                                                                                                                                                                                                                         \\
      = & \frac{1}{N}\sum_{i=1}^N\boldsymbol{Z}_{it}'\left(\widehat{\boldsymbol{\theta}}_{t,k}^1-\boldsymbol{\theta}_{t,k}^1\right)-\frac{1}{N}\sum_{i=1}^N\boldsymbol{Z}_{it}'\left(\widehat{\boldsymbol{\theta}}_{t,k}^0-\boldsymbol{\theta}_{t,k}^0\right)-\frac{1}{N}\sum_{i=1}^N\left(e_{it,k}^1-e_{it,k}^0\right) \\
      = & O_p\left(N_1^{-1/2}\right)+O_p\left(N_0^{-1/2}\right)+O_p\left(\left(N_1+N_0\right)^{-1/2}\right)                                                                                                                                                                                                             \\
      = & O_p\left(N_1^{-1/2}\right)+O_p\left(N_0^{-1/2}\right).
    \end{align*}

    Since $\frac{1}{N}\sum_{i=1}^N\tau_{it,k}-\mathbb{E}\left(\tau_{it,k}\right)=O_p\left(\left(N_1+N_0\right)^{-1/2}\right)$.
    We have that $\frac{1}{N}\sum_{i=1}^N\widehat{\tau}_{it,k}-\tau_{t,k}=\frac{1}{N}\sum_{i=1}^N\widehat{\tau}_{it,k}-\frac{1}{N}\sum_{i=1}^N\tau_{it,k}+\frac{1}{N}\sum_{i=1}^N\tau_{it,k}-\tau_{t,k}=O_p\left(N_1^{-1/2}\right)+O_p\left(N_0^{-1/2}\right)$.

  \end{proof}

  \begin{proof}[Proof of Proposition \ref{Ch3_prop_OLS}]

    (i)
    \begin{align*}
      \tilde{\tau}_{it,k}-\tau_{it,k}= & \boldsymbol{Z}_{it}'\left(\widehat{\boldsymbol{\theta}}_{t,k}^{*1}-\boldsymbol{\theta}_{t,k}^{*1}\right)-\boldsymbol{Z}_{it}'\left(\widehat{\boldsymbol{\theta}}_{t,k}^{*0}-\boldsymbol{\theta}_{t,k}^{*0}\right)-\left(u_{it,k}^1-u_{it,k}^0\right).
    \end{align*}

    Since $\mathbb{E}\left(u_{it,k}^1\mid\boldsymbol{Z}_{it}\right)=0$ and under Assumption \ref{Ch3_assume_error}, we have
    $\mathbb{E}\left(\widehat{\boldsymbol{\theta}}_{t,k}^{*1}-\boldsymbol{\theta}_{t,k}^{*1}\mid\boldsymbol{Z}_t\right)=0$, and
    $\mathbb{E}\left(\widehat{\boldsymbol{\theta}}_{t,k}^{*1}-\boldsymbol{\theta}_{t,k}^{*1}\mid\boldsymbol{Z}_{it}\right)=\mathbb{E}_{-i}\left[\mathbb{E}\left(\widehat{\boldsymbol{\theta}}_{t,k}^{*1}-\boldsymbol{\theta}_{t,k}^{*1}\mid\boldsymbol{Z}_t\right)\right]=0$, where $\mathbb{E}_{-i}(\cdot)$ denotes the expectation taken with respect to $\boldsymbol{Z}_{jt},\ j\neq i$.
    Similarly, we have $\mathbb{E}\left(\widehat{\boldsymbol{\theta}}_{t,k}^{*0}-\boldsymbol{\theta}_{t,k}^{*0}\mid\boldsymbol{Z}_{it}\right)=0$.

    Thus, $\mathbb{E}\left(\tilde{\tau}_{it,k}-\tau_{it,k}\mid \boldsymbol{Z}_{it}=\boldsymbol{z}_{it}\right)=0$. It follows that $\mathbb{E}\left(\tilde{\tau}_{t,k}-\tau_{t,k}\right)=0$ using the law of iterated expectations.

    (ii)
    \begin{align*}
        & \frac{1}{N}\sum_{i=1}^N\left(\tilde{\tau}_{it,k}-\tau_{it,k}\right)                                                                                                                                                                                                                                                                                                                                                                                                        \\
      = & \frac{1}{N}\sum_{i=1}^N\left[\boldsymbol{Z}_{it}'\left(\widehat{\boldsymbol{\theta}}_{t,k}^{*1}-\boldsymbol{\theta}_{t,k}^1\right)+{\boldsymbol{\gamma}_{t,k}^1}'\boldsymbol{\varepsilon}_i^\mathcal{P}-\varepsilon_{it,k}^1\right]  - \frac{1}{N}\sum_{i=1}^N\left[\boldsymbol{Z}_{it}'\left(\widehat{\boldsymbol{\theta}}_{t,k}^{*0}-\boldsymbol{\theta}_{t,k}^0\right)+{\boldsymbol{\gamma}_{t,k}^0}'\boldsymbol{\varepsilon}_i^\mathcal{P}-\varepsilon_{it,k}^0\right] \\
      = & \frac{1}{N}\sum_{i=1}^N\left[\boldsymbol{Z}_{it}'\left(\sum_{j\in\mathcal{T}}\boldsymbol{Z}_{jt}\boldsymbol{Z}_{jt}'\right)^{-1}\sum_{j\in\mathcal{T}}\boldsymbol{Z}_{jt}\left(\varepsilon_{jt,k}^1-{\boldsymbol{\gamma}_{t,k}^1}'\boldsymbol{\varepsilon}_j^\mathcal{P}\right)+{\boldsymbol{\gamma}_{t,k}^1}'\boldsymbol{\varepsilon}_i^\mathcal{P}-\varepsilon_{it,k}^1\right]                                                                                           \\
        & -\frac{1}{N}\sum_{i=1}^N\left[\boldsymbol{Z}_{it}'\left(\sum_{j\in\mathcal{C}}\boldsymbol{Z}_{jt}\boldsymbol{Z}_{jt}'\right)^{-1}\sum_{j\in\mathcal{C}}\boldsymbol{Z}_{jt}\left(\varepsilon_{jt,k}^0-{\boldsymbol{\gamma}_{t,k}^0}'\boldsymbol{\varepsilon}_j^\mathcal{P}\right)+{\boldsymbol{\gamma}_{t,k}^0}'\boldsymbol{\varepsilon}_i^\mathcal{P}-\varepsilon_{it,k}^0\right].
    \end{align*}

    The following two statements hold:
    \begin{align*}
      \frac{1}{N}\sum_{i=1}^N\boldsymbol{Z}_{it}'\left(\sum_{j\in\mathcal{T}}\boldsymbol{Z}_{jt}\boldsymbol{Z}_{jt}'\right)^{-1}\sum_{j\in\mathcal{T}}\boldsymbol{Z}_{jt}\varepsilon_{jt,k}^1 & =O_p\left(N_1^{-1/2}\right), \\
      \frac{1}{N}\sum_{i=1}^N\boldsymbol{Z}_{it}'\left(\sum_{j\in\mathcal{C}}\boldsymbol{Z}_{jt}\boldsymbol{Z}_{jt}'\right)^{-1}\sum_{j\in\mathcal{C}}\boldsymbol{Z}_{jt}\varepsilon_{jt,k}^0 & =O_p\left(N_0^{-1/2}\right).
    \end{align*}

    Following similar arguments as in \cite{li2017estimation}, we denote $\tilde{\Delta}_i^1={\boldsymbol{\gamma}_{t,k}^1}'\boldsymbol{\varepsilon}_i^\mathcal{P}-\varepsilon_{it,k}^1-\boldsymbol{Z}_{it}'\left(\sum_{j\in\mathcal{T}}\boldsymbol{Z}_{jt}\boldsymbol{Z}_{jt}'\right)^{-1}\sum_{j\in\mathcal{T}}\boldsymbol{Z}_{jt}{\boldsymbol{\gamma}_{t,k}^1}'\boldsymbol{\varepsilon}_j^\mathcal{P}$, and $\Delta_i^1={\boldsymbol{\gamma}_{t,k}^1}'\boldsymbol{\varepsilon}_i^\mathcal{P}-\varepsilon_{it,k}^1-\boldsymbol{Z}_{it}'\mathbb{E}\left(\boldsymbol{Z}_{it}\boldsymbol{Z}_{it}'\right)^{-1}\mathbb{E}\left(\boldsymbol{Z}_{it}{\boldsymbol{\gamma}_{t,k}^1}'\boldsymbol{\varepsilon}_i^\mathcal{P}\right)$.
    We have that $\mathbb{E}\left(\boldsymbol{Z}_{it}\Delta_i^1\right)=0$.
    Since $\boldsymbol{Z}_{it}$ contains constant 1, it follows that
    $\mathbb{E}\left(\Delta_i^1\right)=0$.
    Thus $\frac{1}{N}\sum_{i=1}^N\tilde{\Delta}_i^1\overset{p}{\rightarrow}\mathbb{E}\left(\Delta_i^1\right)=0$ as $N_1\rightarrow\infty$, and $\frac{1}{N}\sum_{i=1}^N\tilde{\Delta}_i^1=O_p\left(N_1^{-1/2}\right)$.
    Similarly, we have $\frac{1}{N}\sum_{i=1}^N\tilde{\Delta}_i^0=O_p\left(N_0^{-1/2}\right)$.

    Thus, $\frac{1}{N}\sum_{i=1}^N\tilde{\tau}_{it,k}-\frac{1}{N}\sum_{i=1}^N\tau_{it,k}=O_p\left(N_1^{-1/2}\right)+O_p\left(N_0^{-1/2}\right)$, and $\frac{1}{N}\sum_{i=1}^N\tilde{\tau}_{it,k}-\tau_{t,k}=O_p\left(N_1^{-1/2}\right)+O_p\left(N_0^{-1/2}\right)$.

  \end{proof}

\end{appendices}

 \bibliography{References3}

\begin{thebibliography}{}

\bibitem[Abadie and Cattaneo, 2018]{abadie2018econometric}
Abadie, A. and Cattaneo, M.~D. (2018).
\newblock Econometric methods for program evaluation.
\newblock {\em Annual Review of Economics}, 10:465--503.

\bibitem[Abadie et~al., 2010]{abadie2010synthetic}
Abadie, A., Diamond, A., and Hainmueller, J. (2010).
\newblock Synthetic control methods for comparative case studies: Estimating
  the effect of california’s tobacco control program.
\newblock {\em Journal of the American Statistical Association},
  105(490):493--505.

\bibitem[Abadie et~al., 2015]{abadie2015comparative}
Abadie, A., Diamond, A., and Hainmueller, J. (2015).
\newblock Comparative politics and the synthetic control method.
\newblock {\em American Journal of Political Science}, 59(2):495--510.

\bibitem[Acemoglu et~al., 2008]{acemoglu2008income}
Acemoglu, D., Johnson, S., Robinson, J.~A., and Yared, P. (2008).
\newblock Income and democracy.
\newblock {\em American Economic Review}, 98(3):808--42.

\bibitem[Ahn et~al., 2013]{ahn2013panel}
Ahn, S.~C., Lee, Y.~H., and Schmidt, P. (2013).
\newblock Panel data models with multiple time-varying individual effects.
\newblock {\em Journal of Econometrics}, 174(1):1--14.

\bibitem[Andrews, 1999]{andrews1999consistent}
Andrews, D.~W. (1999).
\newblock Consistent moment selection procedures for generalized method of
  moments estimation.
\newblock {\em Econometrica}, 67(3):543--563.

\bibitem[Athey et~al., 2021]{athey2021matrix}
Athey, S., Bayati, M., Doudchenko, N., Imbens, G., and Khosravi, K. (2021).
\newblock Matrix completion methods for causal panel data models.
\newblock {\em Journal of the American Statistical Association}, pages 1--15.

\bibitem[Athey and Imbens, 2017]{athey2017state}
Athey, S. and Imbens, G.~W. (2017).
\newblock The state of applied econometrics: Causality and policy evaluation.
\newblock {\em Journal of Economic Perspectives}, 31(2):3--32.

\bibitem[Bai, 2009]{bai2009panel}
Bai, J. (2009).
\newblock Panel data models with interactive fixed effects.
\newblock {\em Econometrica}, 77(4):1229--1279.

\bibitem[Bai and Ng, 2002]{bai2002determining}
Bai, J. and Ng, S. (2002).
\newblock Determining the number of factors in approximate factor models.
\newblock {\em Econometrica}, 70(1):191--221.

\bibitem[Berry et~al., 1995]{berry1995automobile}
Berry, S., Levinsohn, J., and Pakes, A. (1995).
\newblock Automobile prices in market equilibrium.
\newblock {\em Econometrica: Journal of the Econometric Society}, pages
  841--890.

\bibitem[Ferman and Pinto, 2019]{ferman2019synthetic}
Ferman, B. and Pinto, C. (2019).
\newblock Synthetic controls with imperfect pre-treatment fit.
\newblock {\em arXiv preprint arXiv:1911.08521}.

\bibitem[Finkelstein et~al., 2012]{finkelstein2012oregon}
Finkelstein, A., Taubman, S., Wright, B., Bernstein, M., Gruber, J., Newhouse,
  J.~P., Allen, H., Baicker, K., and {Oregon Health Study Group} (2012).
\newblock The {Oregon Health Insurance Experiment}: evidence from the first
  year.
\newblock {\em The Quarterly Journal of Economics}, 127(3):1057--1106.

\bibitem[Hansen, 2021]{hansen2021econometrics}
Hansen, B.~E. (2021).
\newblock Econometrics.
\newblock {\em Manuscript}.
\newblock https://www.ssc.wisc.edu/~bhansen/econometrics/.

\bibitem[Holtz-Eakin et~al., 1988]{holtz1988estimating}
Holtz-Eakin, D., Newey, W., and Rosen, H.~S. (1988).
\newblock Estimating vector autoregressions with panel data.
\newblock {\em Econometrica: Journal of the econometric society}, pages
  1371--1395.

\bibitem[Hsiao et~al., 2012]{hsiao2012panel}
Hsiao, C., Steve~Ching, H., and Ki~Wan, S. (2012).
\newblock A panel data approach for program evaluation: measuring the benefits
  of political and economic integration of {Hong Kong} with mainland {China}.
\newblock {\em Journal of Applied Econometrics}, 27(5):705--740.

\bibitem[Li and Bell, 2017]{li2017estimation}
Li, K.~T. and Bell, D.~R. (2017).
\newblock Estimation of average treatment effects with panel data: Asymptotic
  theory and implementation.
\newblock {\em Journal of Econometrics}, 197(1):65--75.

\bibitem[Rosenbaum and Rubin, 1983]{rosenbaum1983central}
Rosenbaum, P.~R. and Rubin, D.~B. (1983).
\newblock The central role of the propensity score in observational studies for
  causal effects.
\newblock {\em Biometrika}, 70(1):41--55.

\bibitem[Rubin, 1974]{rubin1974estimating}
Rubin, D.~B. (1974).
\newblock Estimating causal effects of treatments in randomized and
  nonrandomized studies.
\newblock {\em Journal of Educational Psychology}, 66(5):688.

\bibitem[Taubman et~al., 2014]{taubman2014medicaid}
Taubman, S.~L., Allen, H.~L., Wright, B.~J., Baicker, K., and Finkelstein,
  A.~N. (2014).
\newblock Medicaid increases emergency-department use: evidence from {Oregon's
  Health Insurance Experiment}.
\newblock {\em Science}, 343(6168):263--268.

\bibitem[Xu, 2017]{xu2017generalized}
Xu, Y. (2017).
\newblock Generalized synthetic control method: Causal inference with
  interactive fixed effects models.
\newblock {\em Political Analysis}, 25(1):57--76.

\end{thebibliography}
 \bibliographystyle{apalike}

 \end{document}